\documentclass[onefignum,onetabnum]{siamart190516}

\author{R\'emy Belmonte\thanks{University of Electro-Communications, Chofu, Tokyo
(\email{remybelmonte@gmail.com}).} \and Eun Jung Kim\thanks{Universit\'e
Paris-Dauphine, PSL University, CNRS, LAMSADE, Paris.
(\email{eun-jung.kim@dauphine.fr}, \email{michail.lampis@dauphine.fr}).} \and
Michael Lampis\footnotemark[3] \and Valia Mitsou\thanks{Universit\'e de Paris,
IRIF, CNRS, Paris.(\email{vmitsou@irif.fr})} \and Yota Otachi\thanks{Nagoya
University, Nagoya, 464-8601, Japan(\email{otachi@nagoya-u.jp})}}

\usepackage{subcaption}
\usepackage{amsopn}

\headers{Grundy Distinguishes Treewidth from Pathwidth}{R. Belmonte, E.J. Kim, M. Lampis, V. Mitsou, Y. Otachi}

\title{Grundy Distinguishes Treewidth from Pathwidth\thanks{Accepted for publication in SIDMA. A conference version of this paper appeared in ESA 2020.
\funding{All authors supported under the PRC CNRS JSPS 2019-2020 program, project PARAGA (Parameterized Approximation Graph Algorithms). R\'emy Belmonte was partially supported by JSPS KAKENHI Grant Number JP18K11157. Eun Jung Kim and Michael Lampis were partially supported by ANR JCJC Grant Number 18-CE40-0025-01 (ASSK) and 21-CE48-0022 (S-EX-AP-PE-AL). Yota Otachi was partially supported by JSPS KAKENHI Grant Numbers JP18K11168, JP18K11169, JP18H04091.}}}

\usepackage{amsfonts}
\usepackage{graphicx}
\usepackage{epstopdf}
\usepackage{algorithmic}
\ifpdf
  \DeclareGraphicsExtensions{.eps,.pdf,.png,.jpg}
\else
  \DeclareGraphicsExtensions{.eps}
\fi

\newsiamremark{remark}{Remark}
\newsiamremark{hypothesis}{Hypothesis}
\crefname{hypothesis}{Hypothesis}{Hypotheses}
\newsiamthm{claim}{Claim}

\newcommand\grundy{\textsc{Grundy Coloring}~}

\begin{document}

\maketitle

\begin{abstract}

Structural graph parameters, such as treewidth, pathwidth, and clique-width,
are a central topic of study in parameterized complexity. A main aim of
research in this area is to understand the ``price of generality'' of these
widths: as we transition from more restrictive to more general notions, which
are the problems that see their complexity status deteriorate from
fixed-parameter tractable to intractable?  This type of question is by now very
well-studied, but, somewhat strikingly, the algorithmic frontier between the
two (arguably) most central width notions, treewidth and pathwidth, is still
not understood: currently, no natural graph problem is known to be W-hard for
one but FPT for the other. Indeed, a surprising development of the last few
years has been the observation that for many of the most paradigmatic problems,
their complexities for the two parameters actually coincide exactly, despite
the fact that treewidth is a much more general parameter. It would thus appear
that the extra generality of treewidth over pathwidth  often comes ``for
free''.

Our main contribution in this paper is to uncover the first natural example
where this generality comes with a high price. We consider \textsc{Grundy
Coloring}, a variation of coloring where one seeks to calculate the worst
possible coloring that could be assigned to a graph by a greedy First-Fit
algorithm.  We show that this well-studied problem is FPT parameterized by
pathwidth; however, it becomes significantly harder (W[1]-hard) when
parameterized by treewidth.  Furthermore, we show that \textsc{Grundy Coloring}
makes a second complexity jump for more general widths, as it becomes
paraNP-hard for clique-width.  Hence, \textsc{Grundy Coloring} nicely captures
the complexity trade-offs between the three most well-studied parameters.
Completing the picture, we show that \textsc{Grundy Coloring} is FPT
parameterized by modular-width.

\end{abstract}

\begin{keywords}
Treewidth, Pathwidth, Clique-width, Grundy Coloring
\end{keywords}

\section{Introduction}

The study of the algorithmic properties of \emph{structural graph parameters}
has been one of the most vibrant research areas of parameterized complexity in
the last few years. In this area we consider graph complexity measures (``graph
width parameters''), such as treewidth, and attempt to characterize the class
of problems which become tractable for each notion of width.  The most
important graph widths are often comparable to each other in terms of their
generality.  Hence, one of the main goals of this area is to understand which
problems separate two comparable parameters, that is, which problems transition
from being FPT for a more restrictive parameter to W-hard for a more general
one\footnote{We assume the reader is familiar with the basics of parameterized
complexity theory, such as the classes FPT and W[1], as given in standard
textbooks \cite{CyganFKLMPPS15}.}.  This endeavor is sometimes referred to as
determining the ``price of generality'' of the more general parameter. 

Treewidth and pathwidth, which have an obvious containment relationship to each
other, are possibly the two most well-studied graph width parameters.  Despite
this, to the best of our knowledge, no natural problem is currently known to
delineate their complexity border in the sense we just described. Our main
contribution is exactly to uncover a natural, well-known problem which fills
this gap.  Specifically, we show that \textsc{Grundy Coloring}, the problem of
ordering the vertices of a graph to maximize the number of colors used by the
First-Fit coloring algorithm, is FPT parameterized by pathwidth, but W[1]-hard
parameterized by treewidth. We then show that \textsc{Grundy Coloring} makes a
further complexity jump if one considers clique-width, 
%which is arguably the next most well-studied parameter, 
as in this case the problem
is paraNP-complete. Hence, \textsc{Grundy Coloring} turns out to be an 
interesting specimen, nicely demonstrating the algorithmic
trade-offs involved among the three most central graph widths.

\subparagraph*{Graph widths and the price of generality}  Much of modern
parameterized complexity theory is centered around studying %the %algorithmic
%properties % and trade-offs 
%of 
graph widths, especially treewidth and its
variants.  In this paper we focus on the parameters %whose relations are
summarized in Figure \ref{fig:params}, and especially the parameters that form
a linear hierarchy, from vertex cover, to tree-depth, pathwidth, treewidth, and
clique-width. Each of these parameters is %known to be 
a strict generalization
of the previous ones in this list. On the algorithmic level we would expect
this relation to manifest itself by the appearance of more and more problems
which become \emph{intractable} as we move towards the more general parameters.
Indeed, a search through the literature reveals that for each step in
this list of parameters, %by now 
several \emph{natural} problems have been
discovered which distinguish the two consecutive parameters (we give more
details below)%, exactly as we would expect
.  The one glaring exception to this
rule seems to be the relation between treewidth and pathwidth.

Treewidth is a parameter of central importance to parameterized algorithmics,
in part because wide classes of problems (notably all MSO$_2$-expressible
problems \cite{Courcelle90}) are FPT for this parameter. Treewidth is usually
defined in terms of tree decompositions of graphs, which naturally leads to the
equally well-known notion of pathwidth, defined by forcing the decomposition to
be a path. On a graph-theoretic level, the difference between the two notions
is well-understood and treewidth is known to describe a much richer class of
graphs. In particular, while all graphs of pathwidth $k$ have treewidth at most
$k$, there exist graphs of constant treewidth (in fact, even trees) of
unbounded pathwidth. Naturally, one would expect this added richness of
treewidth to come with some negative algorithmic consequences in the form of
problems which are FPT for pathwidth but W-hard for treewidth.  Furthermore,
since treewidth and pathwidth are probably the most studied parameters in our
list, one might expect the problems that distinguish the two to be the first
ones to be discovered.

Nevertheless, so far this (surprisingly) does not seem to have been the case: on the one hand, FPT algorithms for pathwidth are DPs which also extend to treewidth; on the other hand, %to the best of our knowledge, every \emph{natural} problem that is currently known to be in XP, but W-hard parameterized by treewidth, is in fact also W-hard for pathwidth.   %and pathwidth appear to be the only two parameters on our list for which no natural problem is known to be FPT for one but W[1]-hard for the other. Indeed, this is clearly the case for all problems mentioned above, and we give (in the appendix) a
%Indeed, 
we give (in Section~\ref{sec:list}) a semi-exhaustive list of dozens of natural %additional 
problems which are W[1]-hard for
treewidth and turn out without exception to also be hard for pathwidth. %We observe that, 
In fact, even when this is sometimes not explicitly stated in the
literature, the same reduction that establishes W-hardness by treewidth also does so for pathwidth. Intuitively, an explanation for
this phenomenon is that the basic structure of such reductions typically
resembles a $k\times n$ (or smaller) grid, which has both
treewidth and pathwidth bounded by $k$.  %Hence, to the best of our knowledge, \textsc{Grundy Coloring} appears to be the first problem distinguishing these two parameters.

Our main motivation in this paper is to take a closer
look at the algorithmic barrier between pathwidth and treewidth and try to
locate a natural (that is, not artificially contrived) problem whose
complexity transitions from FPT to W-hard at this barrier.  Our main result is
the proof that \textsc{Grundy Coloring} is such a problem.  This puts in the
picture the last missing piece of the puzzle, as we now have natural problems
that distinguish the parameterized complexity of any two consecutive parameters
in our main hierarchy.

%which, together with the currently known problems
%that separate vertex cover from tree-depth, tree-depth from pathwidth, and
%treewidth from clique-width, gives a better idea of the price of generality
%paid in each transition in this main hierarchy of structural parameters.

\begin{figure}[htb]

\begin{tabular}{lr}
\begin{minipage}{0.34\textwidth}
\includegraphics[scale=0.75]{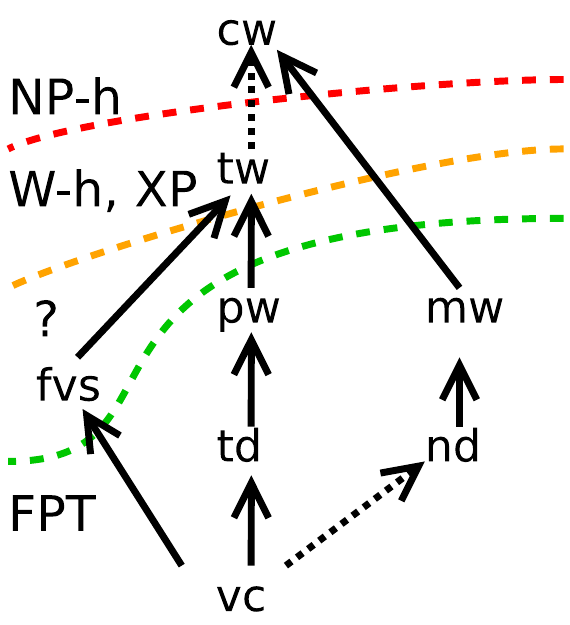}
\end{minipage} &

\begin{minipage}{0.59\textwidth}

\begin{tabular}{lll}
%\hline
%Summary of results\\
\hline
Parameter & Result & Ref \hfill \\
\hline
Clique-width & paraNP-hard & Thm~\ref{thm:cw} \\
Treewidth & W[1]-hard & Thm~\ref{thm:tw} \\
Pathwidth & FPT & Thm~\ref{thm:pw} \\
Modular-width & FPT & Thm~\ref{thm:mw} \\
\hline
\end{tabular}

\medskip

In the figure, clique-width, treewidth, pathwidth, tree-depth, vertex cover,
feedback vertex set, neighborhood diversity, and modular-width are indicated as
cw, tw, pw, td, vc, fvs, nd, and mw respectively. Arrows indicate more general
parameters. Dotted arrows indicate that the parameter may increase
exponentially, (e.g. graphs of vc $k$ have nd at most $2^k+k$).

\end{minipage}
\end{tabular}
\caption{Summary of considered graph parameters and results. }
\label{fig:params} \end{figure}

\subparagraph*{Grundy Coloring} In the \textsc{Grundy Coloring} problem we are
given a graph $G=(V,E)$ and are asked to order %the vertices of 
$V$ in a way that maximizes the number of colors used by the greedy (First-Fit) coloring algorithm. 
%An equivalent definition
%can be given in terms of colorings with a domination relation between color
%classes (see Definition \ref{def:grundy}). 
The notion of Grundy coloring was
first introduced by Grundy in the 1930s, and later formalized in
\cite{ChristenS79}.  Since then, the complexity of \textsc{Grundy Coloring} has
been very well-studied (see
\cite{ABKS20,AraujoS12,BonnetLP18,ErdosHLP03,GyarfasL88,HavetS13,KiersteadS11,Kortsarz07,TangWHZ17,TelleP97,Zaker05,Zaker06,Zaker07a}
and the references therein).  For the natural parameter, namely the number of
colors to be used, Grundy coloring was recently proved to be W[1]-hard
in~\cite{ABKS20}.  An XP algorithm for \textsc{Grundy Coloring} parameterized
by treewidth was given in \cite{TelleP97}, using the fact that the Grundy
number of any graph is at most $\log n$ times its treewidth.  In
\cite{BonnetFKS18} Bonnet et al.\ explicitly asked whether this can be improved
to an FPT algorithm. They also observed that the problem is FPT parameterized
by vertex cover. It appears that the complexity of \textsc{Grundy Coloring}
parameterized by pathwidth was never explicitly posed as a question and it was
not suspected that it may differ from that for treewidth. We note that, since
the problem  can be seen to be MSO$_1$-expressible for a fixed Grundy number
(indeed in Definition \ref{def:grundy} we reformulate it as a coloring problem
where each color class dominates later classes, which is an MSO$_1$-expressible
property), it is FPT for all considered parameters if the Grundy number is also
a parameter \cite{CourcelleMR00}, so we intuitively want to concentrate on
cases where the Grundy number is large.

\subparagraph*{Our results} Our results illuminate the complexity of \textsc{Grundy Coloring}
parameterized by pathwidth and treewidth, as well as clique-width and
modular-width. More specifically:

\begin{enumerate}

\item We show that \textsc{Grundy Coloring} is W[1]-hard parameterized by treewidth via a
reduction from $k$-\textsc{Multi-Colored Clique}. The main building block of our reduction is the structure of binomial trees, which have treewidth one but unbounded pathwidth, which explains the complexity jump between the two parameters. As mentioned, an XP algorithm
is known in this case \cite{TelleP97}, so this result is in a sense tight.

\item We observe that \textsc{Grundy Coloring} is FPT parameterized by
pathwidth. Our main tool here is a combinatorial lemma stating that on any
graph the Grundy number is at most a linear function of the pathwidth, which
was first shown in \cite{DujmovicJW12}, using previous results on the First-Fit
coloring of interval graphs \cite{KST16,NB08}. To obtain an FPT algorithm we
simply combine this lemma with the algorithm of \cite{TelleP97}.

\item We show that \textsc{Grundy Coloring} is paraNP-complete parameterized by clique-width,
that is, NP-complete for graphs of constant clique-width (specifically,
clique-width $8$).

\item We show that \textsc{Grundy Coloring} is FPT parameterized by neighborhood diversity
(which is defined in~\cite{Lampis12}) and leverage this result to obtain an FPT algorithm parameterized by
modular-width (which is defined in~\cite{GajarskyLO13}).
\end{enumerate}

Our main interest is concentrated in the first two results, which
achieve our goal of finding a natural problem distinguishing pathwidth from
treewidth. The result for clique-width nicely fills out the picture by
giving an intuitive view of the evolution of the complexity of the problem and
showing that in a case where no non-trivial bound can be shown on the optimal
value, the problem becomes hopelessly hard from the parameterized point of
view. 

\subparagraph*{Other related work} 
Let us now give a brief survey of ``price
of generality'' results involving our considered parameters, that is, results
showing that a problem is efficient for one parameter but hard for a more general
one.  In this area, the results of Fomin et al. \cite{FominGLS10}, introducing the term ``price of generality'', have been particularly impactful.
This work and its follow-ups \cite{FominGLS14,FominGLSZ19}, were the first to
show that four natural graph problems (\textsc{Coloring, Edge Dominating Set,
Max Cut, Hamiltonicity}) which are FPT for treewidth, become W[1]-hard for
clique-width. In this sense, these problems, as well as problems discovered
later such as counting perfect matchings \cite{CurticapeanM16}, \textsc{SAT}
\cite{OrdyniakPS13,DellKLMM17}, $\exists\forall$-\textsc{SAT} \cite{LampisM17},
\textsc{Orientable Deletion} \cite{HanakaKLOS18}, and $d$-\textsc{Regular
Induced Subgraph} \cite{BroersmaGP13}, form part of the ``price'' we have to
pay for considering a more general parameter. This line of research has thus
helped to illuminate the complexity border between the two most important
sparse and dense parameters (treewidth and clique-width), by giving a list of
\emph{natural} problems distinguishing the two. (An artificial
MSO$_2$-expressible such problem was already known much earlier
\cite{CourcelleMR00,abs-1302-4266}).

%Among these problems, it is worthwhile to single out \textsc{Coloring} because
%it displays a complexity behavior analogous to \textsc{Grundy Coloring}, but at
%one level higher in the hierarchy. Namely, not only does it transition from FPT
%for treewidth to W-hard for clique-width, but crucially, its tractability for
%treewidth is due to a combinatorial lemma that bounds the number of colors by a
%function of the width.
%, just as is the case for \textsc{Grundy Coloring} and pathwidth.

Let us now focus in the area below treewidth in Figure \ref{fig:params} by
considering problems which are in XP but W[1]-hard parameterized by treewidth.
By now, there is a small number of problems in this category which are known to
be W[1]-hard even for vertex cover: \textsc{List Coloring}
\cite{FellowsFLRSST11} was the first such problem, followed by \textsc{CSP}
(for the vertex cover of the dual graph) \cite{SamerS10}, and more recently by
$(k,r)$-\textsc{Center}, $d$-\textsc{Scattered Set}, and \textsc{Min Power
Steiner Tree} \cite{KatsikarelisLP19,KatsikarelisLP18,KaurM20} on weighted
graphs. Intuitively, it is not surprising that problems W[1]-hard parameterized
by vertex cover are few and far between, since this is a very restricted
parameter.  Indeed, for most problems in the literature which are W[1]-hard by
treewidth, vertex cover is the only parameter (among the ones considered here)
for which the problem becomes FPT.

A second interesting category are problems which are FPT for tree-depth
(\cite{NesetrilM06}) but W[1]-hard for pathwidth. 
%This is a category which could in principle contain many problems, but until
%recently few examples were known. 
\textsc{Mixed Chinese Postman Problem} was the first discovered problem of this
type \cite{GutinJW16}, followed by \textsc{Min Bounded-Length Cut}
\cite{DvorakK18,bentert2019lengthbounded}, \textsc{ILP} 
%(parameterized also by the maximum co-efficient) 
\cite{GanianO18}, \textsc{Geodetic Set} \cite{kellerhals2020parameterized} and
unweighted $(k,r)$-\textsc{Center} and  $d$-\textsc{Scattered Set}
\cite{KatsikarelisLP19,KatsikarelisLP18}. 
% on unweighted graphs.
More recently, $(A,\ell)$-\textsc{Path Packing} was also shown to belong in this category \cite{BelmonteHKK0KLO20}.

\smallskip

To the best of our knowledge, for all remaining problems which are known to be
W[1]-hard by treewidth, the reductions that exist in the literature also establish 
W[1]-hardness for pathwidth. Below we give a (semi-exhaustive) list of problems 
which are known to be W[1]-hard by treewidth. After reviewing 
the relevant works we have verified that all of the following problems are in 
fact shown to be W[1]-hard parameterized by pathwidth (and in many case by 
feedback vertex set and tree-depth), even if this is not explicitly claimed.

\subsection{Known problems which are W-hard for treewidth and for
pathwidth}\label{sec:list}

\begin{itemize}

\item \textsc{Precoloring Extension} and \textsc{Equitable Coloring} are shown
to be W[1]-hard for both tree-depth and feedback vertex set in
\cite{FellowsFLRSST11} (though the result is claimed only for treewidth). This
is important, because \textsc{Equitable Coloring} often serves as a starting
point for reductions to other problems. A second hardness proof for this
problem was recently given in \cite{GomesLS19}. These two problems are FPT by
vertex cover \cite{FialaGK11}.

\item \textsc{Capacitated Dominating Set} and \textsc{Capacitated Vertex Cover}
are W[1]-hard for both tree-depth and feedback vertex set \cite{DomLSV08}
(though again the result is claimed for treewidth).  

\item  \textsc{Min Maximum Out-degree} on weighted graphs is W[1]-hard by
tree-depth and feedback vertex set \cite{abs-1107-1177}.

\item \textsc{General Factors} is W[1]-hard by tree-depth and feedback vertex set \cite{SamerS11}.

\item \textsc{Target Set Selection} is W[1]-hard by tree-depth and feedback vertex set \cite{Ben-ZwiHLN11} but FPT for vertex cover \cite{NichterleinNUW13}.

\item \textsc{Bounded Degree Deletion} is W[1]-hard by tree-depth and feedback vertex set, but FPT for vertex cover \cite{BetzlerBNU12,GanianKO18}.

\item \textsc{Fair Vertex Cover} is W[1]-hard by tree-depth and feedback vertex set \cite{KnopMT19}.

\item \textsc{Fixing Corrupted Colorings} is W[1]-hard by tree-depth and feedback vertex set \cite{BiasiL19} (reduction from \textsc{Precoloring Extension}).

\item \textsc{Max Node Disjoint Paths} is W[1]-hard by tree-depth and feedback
vertex set \cite{EneMPR16,FleszarMS16}.

\item \textsc{Defective Coloring} is W[1]-hard by tree-depth and feedback vertex set \cite{BelmonteLM18}.

\item \textsc{Power Vertex Cover} is W[1]-hard by tree-depth but open for feedback vertex set \cite{AngelBEL18}.

\item \textsc{Majority CSP} is W[1]-hard parameterized by the tree-depth of the incidence graph \cite{DellKLMM17}.

\item \textsc{List Hamiltonian Path} is W[1]-hard for pathwidth \cite{MeeksS16}.

\item \textsc{L(1,1)-Coloring} is W[1]-hard for pathwidth, FPT for vertex cover \cite{FialaGK11}.

\item \textsc{Counting Linear Extensions} of a poset is W[1]-hard (under Turing reductions) for pathwidth \cite{EibenGKO19}.

\item \textsc{Equitable Connected Partition} is W[1]-hard by pathwidth and feedback vertex set, FPT by vertex cover \cite{EncisoFGKRS09}.

\item \textsc{Safe Set} is W[1]-hard parameterized by pathwidth, FPT by vertex cover \cite{BelmonteHKLOO19}.

\item \textsc{Matching with Lower Quotas} is W[1]-hard parameterized by pathwidth \cite{ArulselvanCGMM18}.

\item \textsc{Subgraph Isomorphism} is W[1]-hard parameterized by the pathwidth of $G$, even when $G, H$ are connected planar graphs of maximum degree $3$ and $H$ is a tree \cite{MarxP14}.

\item \textsc{Metric Dimension} is W[1]-hard by pathwidth \cite{abs-1907-08093}. This was recently strengthened to paraNP-hardness \cite{abs-2102-09791}, again for pathwidth.

\item \textsc{Simple Comprehensive Activity Selection} is W[1]-hard by pathwidth \cite{EibenGO18}.

\item \textsc{Defensive Stackelberg Game for IGL} is W[1]-hard by pathwidth
(reduction from \textsc{Equitable Coloring}) \cite{0001GLN18}.

\item \textsc{Directed} $(p,q)$-\textsc{Edge Dominating Set} is W[1]-hard
parameterized by pathwidth \cite{BelmonteHK0L18}.

\item \textsc{Maximum Path Coloring} is W[1]-hard for pathwidth \cite{Lampis13}.

\item Unweighted $k$-\textsc{Sparsest Cut} is W[1]-hard parameterized by the
three combined parameters tree-depth, feedback vertex set, and $k$
\cite{javadi2019parameterized}.

\item \textsc{Graph Modularity} is W[1]-hard parameterized by pathwidth plus
feedback vertex set \cite{MeeksS18}.

\item \textsc{Minimum Stable Cut} is W[1]-hard parameterized by pathwidth \cite{Lampis21}.
\end{itemize}

%The above (semi-)exhaustive list shows that, even when this is sometimes not
%explicitly stated in the literature, reductions that establish W-hardness by
%treewidth essentially always also establish hardness for pathwidth, and very
%often also for tree-depth and feedback vertex set. Intuitively, an explanation
%for this phenomenon is that the basic structure of such reductions is typically
%somethings that resembles a $k\times n$ (or smaller) grid, which has both
%treewidth and pathwidth bounded by $k$.

%A survey of the relevant literature reveals that treewidth and pathwidth appear
%to be the only two parameters on our list for which no natural problem is known
%to be FPT for one but W[1]-hard for the other. Indeed, this is clearly the case
%for all problems mentioned above, and we give (in the appendix) a
%semi-exhaustive list of dozens of additional problems which are W[1]-hard for
%treewidth and turn out without exception to also be hard for pathwidth. We
%observe that, even when this is sometimes not explicitly stated in the
%literature, reductions that establish W-hardness by treewidth essentially
%always also establish hardness for pathwidth. Intuitively, an explanation for
%this phenomenon is that the basic structure of such reductions is typically
%somethings that resembles a $k\times n$ (or smaller) grid, which has both
%treewidth and pathwidth bounded by $k$.  Hence, to the best of our knowledge,
%\textsc{Grundy Coloring} appears to be the first problem distinguishing these
%two parameters.

Let us also mention in passing that the algorithmic differences of pathwidth
and treewidth may also be studied in the context of problems which are hard for
constant treewidth. Such problems also generally remain hard for constant pathwidth (examples are
\textsc{Steiner Forest} \cite{Gassner10}, \textsc{Bandwidth}
\cite{monien1986bandwidth}, \textsc{Minimum mcut} \cite{GargVY97}). 
One could also potentially try to distinguish between pathwidth
and treewidth by considering the parameter dependence of a problem that is FPT
for both. Indeed, for a long time the best-known algorithm for
\textsc{Dominating Set} had complexity $3^k$ for pathwidth, but $4^k$ for
treewidth. Nevertheless, the advent of fast subset convolution techniques
\cite{RooijBR09}, together with tight SETH-based lower bounds
\cite{LokshtanovMS18} has, for most problems, shown that the
complexities on the two parameters coincide exactly.

Finally, let us mention a case where pathwidth and treewidth have been shown to
be quite different in a sense similar to our framework. In \cite{Razgon14}
Razgon showed that a CNF can be compiled into an OBDD (Ordered Binary Decision
Diagram) of size FPT in the pathwidth of its incidence graphs, but there exist
formulas that always need OBDDs of size XP in the treewidth. Although this
result does separate the two parameters, it is somewhat adjacent to what we are
looking for, as it does not speak about the complexity of a decision problem,
but rather shows that an OBDD-producing algorithm parameterized by treewidth
would need XP time simply because it would have to produce a huge output in
some cases.

\section{Definitions and Preliminaries}

For non-negative integers $i,j$, we use $[i,j]$ to denote the set $\{ k\ |\
i\le k\le j\}$. Note that if $j<i$, then the set $[i,j]$ is empty. We will also
write simply $[i]$ to denote the set $[1,i]$. 

We give two equivalent definitions of our main problem.

\begin{definition}\label{def:grundy} A $k$-Grundy Coloring of a graph $G=(V,E)$
is a partition of $V$ into $k$ non-empty sets $V_1,\ldots, V_k$ such that: (i)
for each $i\in[k]$ the set $V_i$ induces an independent set; (ii) for each
$i\in[k-1]$ the set $V_i$ dominates the set $\bigcup_{i<j\le k} V_j$. \end{definition}

\begin{definition} A $k$-Grundy Coloring of a graph $G=(V,E)$ is a proper
$k$-coloring $c:V\to[k]$ that results by applying the First-Fit algorithm on an
ordering of $V$; the First-Fit algorithm colors one by one the vertices in the
given ordering, assigning to a vertex the minimum color that is not already 
assigned to one of its preceding neighbors.  \end{definition}

The Grundy number of a graph $G$, denoted by $\Gamma(G)$, is the maximum $k$ 
such that $G$ admits a $k$-Grundy Coloring. 
In a given Grundy Coloring, if $u\in V_i$ (equiv.\ if $c(u) = i$) we will say that $u$ was given color
$i$. The \textsc{Grundy Coloring} problem is the problem of determining the
maximum $k$ for which a graph $G$ admits a $k$-Grundy Coloring. It is not hard
to see that a proper coloring is a Grundy coloring if and only if every vertex
assigned color $i$ has at least one neighbor assigned color $j$, for each
$j<i$.

\section{W[1]-Hardness for Treewidth}\label{sec:tw}

In this section we prove  that \textsc{Grundy Coloring} parameterized by
treewidth is W[1]-hard (Theorem~\ref{thm:tw}). Our proof relies on a reduction
from $k$-\textsc{Multi-Colored Clique} and initially establishes W[1]-hardness
for a more general %Grundy--type 
problem where we are given a target color for a set of vertices
(Lemma~\ref{lemma:reduction-TGC}); we then reduce this to \textsc{Grundy
Coloring}.  

An interesting aspect of our reduction is that up until a rather advanced
point, the instance we construct has not only bounded treewidth (which is
necessary for the construction to work), but also bounded pathwidth (see
Lemma~\ref{lemma:smallpw}). This would seem to indicate that we are headed
towards a W[1]-hardness result for \textsc{Grundy Coloring} parameterized by
pathwidth, which would contradict the FPT algorithm of Section \ref{sec:pw}!
This is of course not the case, so it is instructive to ponder why the
reduction fails to work for pathwidth. The reason this happens is that the
final step, which reduces our instance to the plain version of \textsc{Grundy
Coloring} needs to rely on a support operation that ``pre-colors'' some of the
vertices and the gadgets we use to achieve this are trees of unbounded Grundy
number.  The results of Section~\ref{sec:pw} indicate that if these gadgets
have unbounded Grundy number, thay \emph{must} also have unbounded pathwidth,
hence there is a good combinatorial reason why our reduction only works for
treewidth.

Let us now present the different parts of our construction. We will make use of
the structure of binomial trees $T_i$.

\begin{definition} The binomial tree $T_i$ with root $r_i$ is a rooted tree
defined recursively in the following way: $T_1$ consists simply of its root
$r_1$; in order to construct $T_i$ for $i>1$, we construct one copy of $T_j$
for all $j<i$ and a special vertex $r_i$, then we connect $r_j$ with $r_i$.  
An alternative equivalent definition of the binomial tree $T_i$, $i\ge 2$ is 
that we construct two trees $T_{i-1}$ , $T_{i-1}'$, we connect their roots 
$r_{i-1}$, $r'_{i-1}$ and select one of them as the new root $r_i$. 
\end{definition}

\begin{proposition}\label{subtrees} Let $i\ge 2$, $T_i$ be a binomial tree and
$1\le t<i$. There exist $2^{i-t-1}$  binomial trees $T_t$ which are
vertex-disjoint and non-adjacent subtrees in $T_i$, where no $T_t$ contains the
root $r_i$ of $T_i$.  \end{proposition}

\begin{proof}

By induction in $i-t$. For $i-t =1$, $T_i$ indeed contains one $T_{i-1}$ that
does not contain the root $r_i$.  Let it be true that $T_{i-1}$ contains
$2^{i-t-2}$ subtrees $T_t$. Then $T_i$ contains two trees $T_{i-1}$ each of
which contains $2^{i-t-2}$ $T_j$, thus $2^{i-t-1}$ in total.  
\end{proof}

\begin{proposition}\label{bincol} $\Gamma(T_i) \le i$. Furthermore,  for all
$j\le i$ there exists a Grundy coloring which assigns color $j$ to the root of
$T_i$.  \end{proposition}

\begin{proof}The first part is trivial since in any graph $G$ with maximum
degree $\Delta$ we have $\Gamma(G) \le \Delta +1$. In this case $\Gamma(T_i)\le
(i-1)+1 = i$. For the second part, we first prove that there is a Grundy
coloring which assigns color $i$ to the root. This can be proven by strong
induction: if for all $k<i$, there is a Grundy coloring which assigns color $k$
to $r_k$  for all $1\le k\le i-1$, then under this coloring, $r_i$ has at least
one neighbor receiving color $k$ for all $1\le k\le i-1$, so it has to receive
color $i$. To assign to the root a color $j<i$ we observe that if $j=1$ this is
trivial; if $j>1$, we use the fact that (by inductive hypothesis) there is a
coloring that assigns color $j-1$ to $r_{j}$, so in this coloring the root
$r_i$ will take color $j$.  \end{proof}

A Grundy coloring of $T_i$ that assigns color $i$ to $r_i$ is called \emph{optimal}. 
If $r_i$ is assigned color $j<i$ then we call the Grundy coloring \emph{sub-optimal}.

We now define a generalization of the Grundy coloring problem with target
colors and show that it is W[1]-hard parameterized by treewidth. We later
describe how to reduce this problem to \textsc{Grundy Coloring} such that the
treewidth does not increase by a lot.

\begin{definition}[\textsc{Grundy Coloring with Targets}] We are given a graph
$G(V,E)$, an integer $t \in$ \emph{I$\!$N} called the target and a subset
$S\subset V$. (For simplicity we will say that vertices of $S$ have target
$t$.) If $G$ admits a Grundy Coloring which assigns color $t$ to some vertex $s\in S$
we say that, for this coloring, vertex $s$ \emph{achieves its target}. 
If there exists a Grundy Coloring of $G$ which assigns to all vertices of $S$ color 
$t$, then we say that $G$ admits a \emph{Target-achieving Grundy Coloring}. 
\textsc{Grundy Coloring with Targets} is the decision problem associated to the 
question ``given $G, S, t$ as defined above, does $G$ admit a Target-achieving Grundy Coloring ?''.
\end{definition}

We will also make use of the following operation:

\begin{definition}[Tree-support] Given a graph $G = (V,E)$, a vertex $u\in V$
and a set $N$ of positive integers, we define the \emph{tree-support} operation
as follows: (a) for all $i\in N$ we add a copy of $T_i$ in the graph; (b) we
connect $u$ to the root $r_i$ of each of the $T_i$. We say that we \emph{add
supports} $N$ on $u$. The trees $T_i$ will be called the \emph{supporting
trees} or \emph{supports} of $u$. Slightly abusing notation, we also call
\emph{supports} the numbers $i\in N$.  \end{definition} 

Intuitively, the tree-support operation ensures that vertex $u$ may have at
least one neighbor of color $i$ for each $i\in N$ in a Grundy coloring, and
thus increase the color $u$ can take.  Observe that adding supporting trees to
a vertex does not increase the treewidth, but does increase the pathwidth (binomial trees have unbounded pathwidth). %of the graph.

Our reduction is from $k$-\textsc{Multi-Colored Clique}, proven to be W[1]-hard in~\cite{FellowsHRV09}: given a
$k$-multipartite graph $G = (V_1,V_2,\ldots,V_k,E)$, decide if for every
$i\in[k]$ we can pick $u_i\in V_i$ forming a clique, where $k$ is the
parameter. We can also assume that $\forall i\in[k], |V_i| = n$, that $n$ is a
power of 2, and that $V_i=\{v_{i,0},v_{i,1},\ldots,v_{i,n-1}\}$.  Furthermore,
let $|E| = m$. We construct an instance of \textsc{Grundy Coloring with
Targets} $G'=(V',E')$ and $t = 2\log n +4$ (where all logarithms are base two)
using the following gadgets:

\begin{description}

\item[Vertex selection $S_{i,j}$.] See Figure~\ref{fig:VertexSelection}. This
gadget consists of $2\log n$ vertices $S_{i,j}^1 \cup S_{i,j}^2 =
\bigcup_{l\in[\log n]}\{s_{i,j}^{2l-1}\}\cup \bigcup_{l\in[\log
n]}\{s_{i,j}^{2l}\}$, where for each $l\in [\log n]$ we connect vertex
$s_{i,j}^{2l-1}$ to $s_{i,j}^{2l}$ thus forming a matching. Furthermore, for
each $l\in [2,\log n]$, we add supports $[2l -2]$ to vertices $s_{i,j}^{2l-1}$
and $s_{i,j}^{2l}$. Observe that the vertices $s_{i,j}^{2l-1}$ and
$s_{i,j}^{2l}$ together with their supports form a binomial tree $T_{2l}$ 
with either of these vertices as the root.  
We construct $k(m+2)$ gadgets $S_{i,j}$, one for each $i\in[k], j\in[0,m+1]$.

\begin{figure}
\begin{subfigure}[b]{.44\linewidth}
%\centering
\includegraphics[scale=0.33]{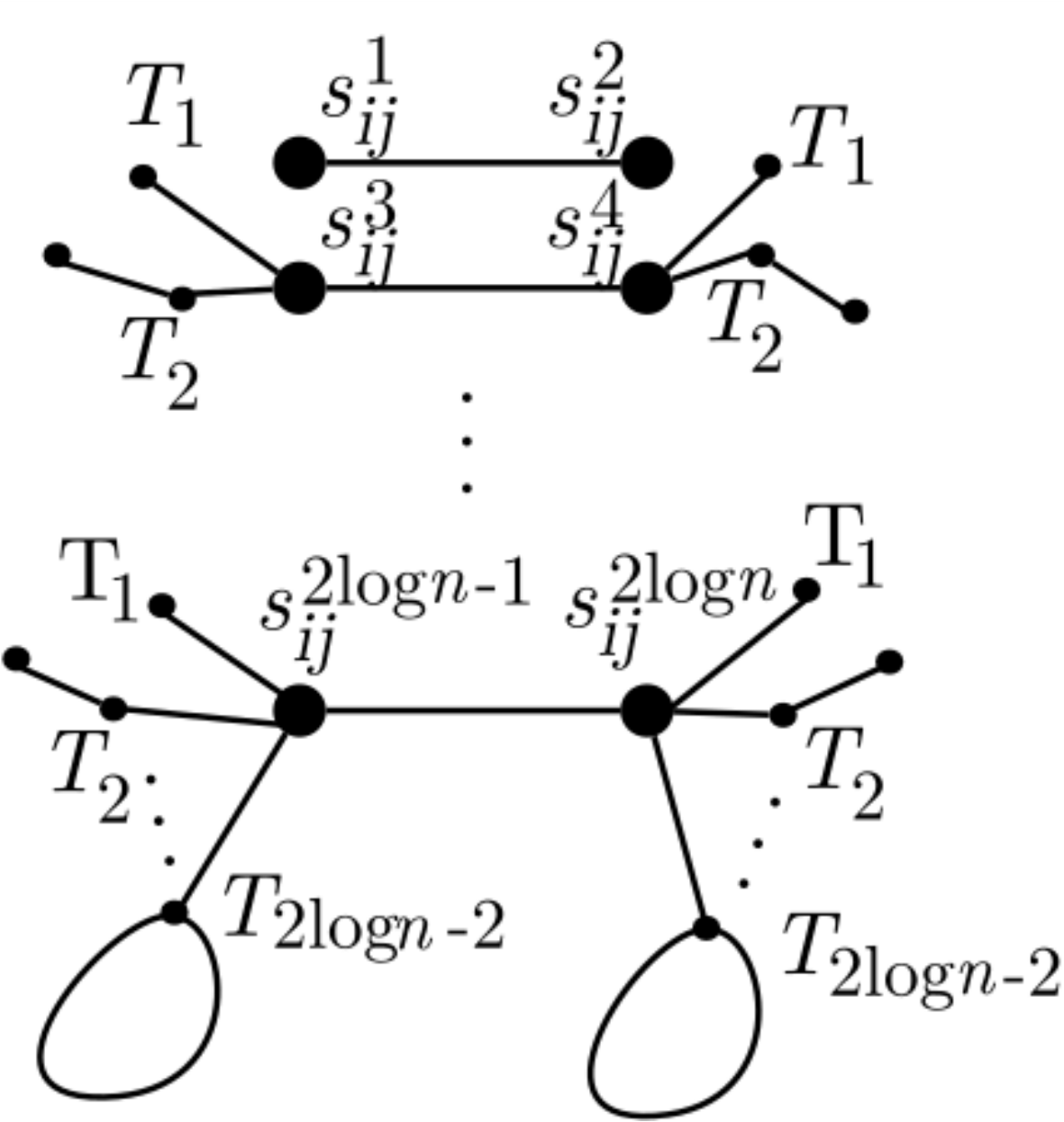}
\caption{Vertex Selection gadget $S_{i,j}$.}
\label{fig:VertexSelection}
\end{subfigure}%
\begin{subfigure}[b]{.56\linewidth}
%\centering
\includegraphics[scale=0.33]{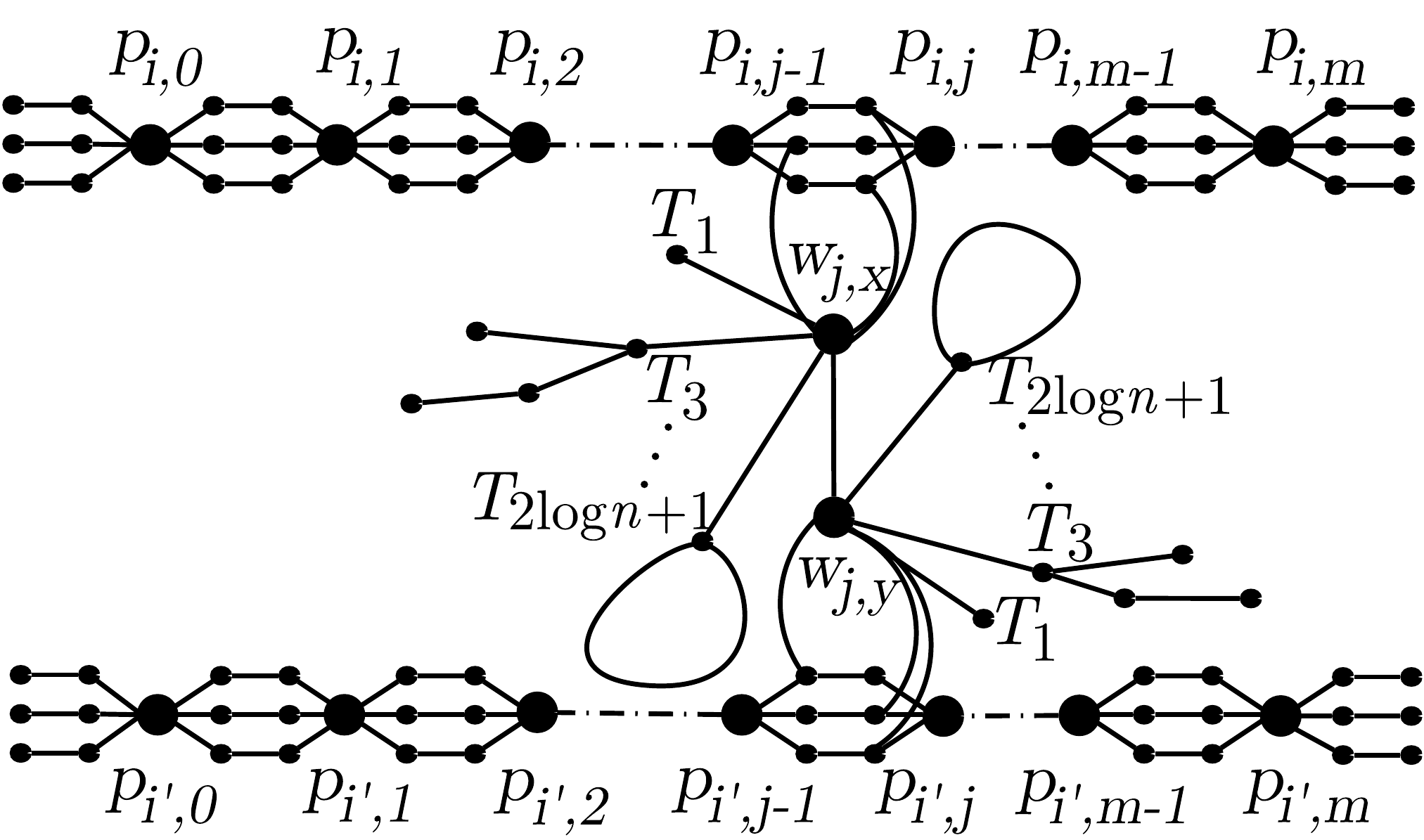} \caption{Propagators $p_{i,j}$ and
Edge Selection gadget $W_{j}$. The edge selection checkers and the supports of the $p_{i,j}$ and $s_{i,j}^l$ are not depicted. In the example $B_x = 010$ and $B_y = 100$.}
\label{fig:EdgeSelection} \end{subfigure} \caption{The gadgets.  Figure\
\ref{fig:VertexSelection} is an enlargement of Figure\ \ref{fig:EdgeSelection}
between $p_{i,j-1}$ and $p_{i,j}$.}\label{fig:Gadgets} \end{figure}

\smallskip

The vertex selection gadget $S_{i,1}$ encodes in binary the vertex that is
selected in the clique from $V_i$. In particular, for each pair
$s_{i,1}^{2l-1},s_{i,1}^{2l}$, $l\in[\log n]$ either of these vertices can take
the maximum color in an optimal Grundy coloring of the binomial tree $T_{2l}$
(that is, a coloring that gives the root of the binomial tree $T_{2l}$ color
$2l$).  A selection corresponds to bit 0 or 1 for the $l^{th}$ binary position.
In order to ensure that for each $j\in [m]$ all (middle) $S_{i,j}$ encode the
same vertex, we use propagators. 

\item[Propagators $p_{i,j}$.] See Figure~\ref{fig:EdgeSelection}. For $i\in[k]$
and $j\in [0,m]$, a propagator $p_{i,j}$ is a single vertex connected to all
vertices of $S_{i,j}^2 \cup S_{i,j+1}^1$. To each $p_{i,j}$, we also add
supports $\{2\log n+1,2\log n+2, 2\log n+3\}$.  The propagators have target
$t=2\log n +4$.

\item[Edge selection $W_j$.] See Figure~\ref{fig:EdgeSelection}. Let $j =
(v_{i,x},v_{i',y})\in E$, where $v_{i,x}\in V_i$ and $v_{i',y} \in V_{i'}$. The
gadget $W_j$ consists of four vertices $w_{j,x}, w_{j,y}, w'_{j,x}, w'_{j,y}$.
We call $w'_{j,x}, w'_{j,y}$ the \emph{edge selection checkers}. We have the
edges $(w_{j,x}$, $w_{j,y}), (w'_{j,x},w_{j,x}), (w'_{j,y},w_{j,y})$.  Let us
now describe the connections of these vertices with the rest of the graph.  Let
$B_x = b_1 b_2 \ldots b_{\log n}$ be the binary representation of $x$. We
connect $w_{j,x}$  to each vertex $s_{ij}^{2l-b_l}$, $l\in [\log n]$ (we do
similarly for $w_{j,y}, S_{i',j}$, and $B_y$).  We add to each of $w_{j,x},
w_{j,y}$ supports $ \bigcup_{l\in [\log n+1]} \{2l-1\}$. We add to each of
$w'_{j,x},w'_{j,y}$ supports $[2\log n+3]\setminus\{2\log n+1\}$ and set the
target $t=2\log n+4$ for these two vertices.  We construct $m$ such gadgets,
one for each edge.  We say that $W_j$ \emph{is activated} if at least one of
$w_{j,x},w_{j,y}$ receives color $2\log n +3$.

\item[Edge validators $q_{i,i'}$.] We construct $k\choose 2$ of these gadgets,
one for each pair $(i,i'), i<i' \in [k]$. The edge validator is a single vertex
that is connected to all vertices $w_{j,x}$ for which $j$ is an edge between
$V_i$ and $V_{i'}$. We add supports $[2\log n+2]$ and a target of $t=2\log n
+4$.

\smallskip

The edge validator plays the role of an ``or'' gadget: in order for it to
achieve its target, at least one of its neighboring edge selection gadgets
should be activated.  
\end{description}

\begin{lemma}\label{lemma:reduction-TGC} $G$ has a clique of size $k$ if and only if $G'$ has a target-achieving Grundy coloring.
\end{lemma}

\begin{proof} $\Rightarrow$) Suppose that $G$ has a clique and we want to
produce a coloring of $G'$. In the remainder, when we say that we color a
support tree ``optimally'', we mean that we color its internal vertices in a
way that gives the root the maximum possible color. 

We color the vertices of $G'$ in the following order: First, we color the
vertex selection gadget $S_{i,j}$.  We start from the supports which we color
optimally.  We then color the matchings as follows: let $v_{i,x}$ be the vertex
that was selected in the clique from $V_i$ and $b_1 b_2 \ldots b_{\log n}$ be
the binary representation of $x$; we color vertices $s_{i,j}^{2l-(1-b_l)}$,
$l\in[\log n]$ with color $2l-1$ and vertices $s_{i,j}^{2l-b_l}$, $l\in[\log
n]$ will receive color $2l$.  For the propagators, we color their supports
optimally.  Propagators have $2\log n +3$ neighbors each, all with different
colors, so they receive color $2\log n +4$, thus achieving the targets. 

Then, we color the edge validators $q_{i,i'}$ and the edge selection gadgets
$W_j$ that correspond to edges of the clique (that is, $j = (v_{i,x},
v_{i',y})\in E$ and $v_{i,x}\in V_i$, $v_{i',y}\in V_{i'}$ are selected in the
clique). We first color the supports of $q_{i,i'}, w_{j,x}, w_{j,y}$ optimally.
From the construction, vertex $w_{j,x}$ is connected with vertices
$s_{i,j}^{2l-b_l}$ which have already been colored $2l$, $l\in [\log n]$ and
with supports $\bigcup_{l\in[\log n+1]}\{2l-1\}$, thus $w_{j,x}$ will receive
color $2\log n +2$. Similarly $w_{j,y}$ already has neighbors which are colored
$[2\log n+1]$, but also $w_{j,x}$, thus it will receive color $2\log n +3$.
These $W_j$ will be activated. Since both $w_{j,x}, w_{j,y}$ connect to
$q_{i,i'}$, the latter will be assigned color $2\log n +4$, thus achieving its
target. As for $w'_{j,x}$ and $w'_{j,y}$, these vertices have one neighbor
colored $c$, where $c=2\log n+2$ or $c=2\log n+3$.  We color their support
$T_c$ sub-optimally so that the root receives color $2\log n+1$; we color their
remaining supports optimally. This way, vertices $w'_{j,x}, w'_{j,y}$ can be
assigned color $t=2\log n+4$, achieving the target.

Finally, for the remaining $W_j$, we claim that we can assign to both $w_{j,x},
w_{j,y}$ a color that is at least as high as $2\log n+1$. Indeed, we assign to
each supporting tree $T_r$ of $w_{j,x}$ a coloring that gives its root the
maximum color that is $\le r$ and does not appear in any neighbor of $w_{j,x}$
in the vertex selection gadget. We claim that in this case $w_{j,x}$ will have
neighbors with all colors in $[2\log n]$, because in every interval $[2l-1,2l]$
for $l\in[\log n]$, $w_{j,x}$ has a neighbor with a color in that interval and
a support tree $T_{2l+1}$. If $w_{j,x}$ has color $2\log n+1$ then we color the
supports of $w'_{j,x}$ optimally and achieve its target, while if $w_{j,x}$ has
color higher than $2\log n+1$, we achieve the target of $w'_{j,x}$ as in the
previous paragraph.

\medskip

$\Leftarrow$) Suppose that $G'$ admits a coloring that achieves the target for
all propagators, edge selection checkers, and edge validators.  We will prove
the following three claims, which together imply the remaining direction of the
lemma: 

\begin{claim}\label{claim:vertex-selection} The coloring of the vertex
selection gadgets is consistent throughout, that is, for each $i\in[k]$ and for
each $j_1,j_2, l$, we have that $s_{i,j_1}^l, s_{i,j_2}^l$ received the same
color.  This coloring corresponds to a selection of $k$ vertices of $G$.
\end{claim}

\begin{claim}\label{claim:edge-selection} $k\choose 2$ edge selection gadgets
have been activated.  They correspond to $k\choose 2$ edges of $G$ being
selected.  \end{claim}

\begin{claim}\label{claim:clique}
If an edge selection gadget $W_j = \{w_{j,x},
w_{j,y}\}$ with $j = (v_{i,x}, v_{i',y})$ has been activated then the coloring
of the vertex selection gadgets $S_{i,j}$ and $S_{i',j}$ corresponds to the
selection of vertices $v_{i,x}$ and $v_{i',y}$. In other words, selected
vertices and edges form a clique of size $k$ in $G$.  \end{claim}

\begin{proof}[Proof of Claim~\ref{claim:vertex-selection}]
Suppose that an edge selection checker $w'_{j,x}$ achieved its target. We
claim that this implies that $w_{j,x}$ has color at least $2\log n+1$. Indeed,
$w'_{j,x}$ has degree $2\log n+3$, so its neighbors must have all distinct
colors in $[2\log n+3]$, but among the supports there are only $2$ neighbors
which may have colors in $[2\log n+1,2\log n+3]$. Therefore, the missing color
must come from $w_{j,x}$.  We now observe that vertices from the vertex
selection gadgets have color at most $2\log n$, because if we exclude from
their neighbors the vertices $w_{j,x}$ (which we argued have color at least
$2\log n+1$) and the propagators (which have target $2\log n+4$), these
vertices have degree at most $2\log n-1$.

Suppose that a propagator $p_{i,j}$ achieves its target of $2 \log n +4$.
Since this vertex has a degree of $2 \log n +3$, that means that all of its
neighbors should receive all the colors in $[2 \log n +3]$. As argued, colors
$[2\log n+1, 2\log n+3]$ must come from the supports. Therefore, the colors
$[2\log n]$ come from the neighbors of $p_{i,j}$ in the vertex selection
gadgets.

We now note that, because of the degrees of vertices in vertex selection
gadgets, only vertices $s_{i,j}^{2\log n},s_{i,j+1}^{2\log n-1}$ can receive
colors $2\log n, 2\log n -1$; from the rest, only $s_{i,j}^{2\log
n-2},s_{i,j+1}^{2\log n-3}$ can receive colors $2\log n-2, 2\log n-3$ etc.
Thus, for each $l\in [\log n]$, if $s_{i,j}^{2l}$ receives color $2l-1$ then
$s_{i,j+1}^{2l-1}$ should receive color $2l$ and vice versa. With similar
reasoning, in all vertex selection gadgets we have that $s_{i,j}^{2l-1},
s_{i,j}^{2l}$ received the two colors
$\{2l-1,2l\}$ since they are neighbors. As a result, the colors of $s_{i,j+1}^{2l-1}$, $s_{i,j}^{2l-1}$ (and thus the colors of $s_{i,j+1}^{2l}$,  $s_{i,j}^{2l}$) are the same, therefore, the coloring is consistent, for all values of
$j\in[m]$. \end{proof}

\begin{proof}[Proof of Claim~\ref{claim:edge-selection}]
If an edge validator achieves its target of $2\log n +4$, then at least one of
its neighbors from an edge selection gadget has received color $2\log n +3$. We
know that each edge selection gadget only connects to a unique edge validator, so
there should be $k \choose 2$ edge selection gadgets which have been activated
in order for all edge validators to achieve the target. 
\end{proof}

\begin{proof}[Proof of Claim~\ref{claim:clique}]
Suppose that an edge validator $q_{i,i'}$ achieves its target. That means that
there exists an edge selection gadget $W_j = \{w_{j,x}, w_{j,y}, w'_{j,x},
w'_{j,y}\}$ for which at least one of its vertices $\{w_{j,x}, w_{j,y}\}$, say vertex $w_{j,x}$, has received
color $2\log n +3$.  Let $j$ be an edge connecting $v_{i,x}\in V_i$ to
$v_{i',y}\in V_{i'}$.  Since the degree of $w_{j,x}$ is $2 \log n +4$ and we
have already assumed that two of its neighbors ($q_{i,i'}$ and $w'_{j,x}$) have color $2\log n +4$, in order
for it to receive color $2\log n +3$ all its other neighbors should receive all
colors in $[2\log n +2]$. The only possible assignment is to give colors $2l$,
$l\in [\log n]$ to its neighbors from $S_{i,j}$ and color $2\log n +2$ to
$w_{j,y}$.  The latter is, in turn, only possible if the neighbors of $w_{j,y}$
from $S_{i',j}$ receive all colors $2l$, $l\in [\log n]$. The above corresponds
to selecting vertex $v_{i,x}$ from $V_i$ and $v_{i',y}$ from $V_{i'}$. 
\end{proof} \end{proof}

\begin{lemma}\label{lemma:smallpw} Let $G''$ be the graph that results from
$G'$ if we remove all the tree-supports. Then $G''$ has pathwidth at most
${k\choose 2} + 2k + 3$.  \end{lemma}

\begin{proof} We will use the equivalent definition of pathwidth as a
node-searching game, where the robber is eager and invisible and the cops are
placed on nodes \cite{Bodlaender98}. We will use ${k\choose 2} + 2k + 4$ cops
to clean $G''$ as follows: We place $k\choose 2$ cops on the edge validators.
Then, starting from $j=0$, we place $2k$ cops on the propagators $p_{i,0},
p_{i,1}$ for $i=1,\ldots, k$, plus 2 cops on the edge selection vertices
$w_{j,x}$, $w_{j,y}$ that correspond to edge $j$. We use the two additional
cops to clean line by line the gadgets $S_{i,j}$. We then use one of these cops
to clear $w'_{j,x}, w'_{j,y}$. We continue then to the next column $j=2$ by
removing the $k$ cops from the propagators $p_{i,1}$ and placing them to
$p_{i,3}$. We continue for $j=3, \ldots m-1$ until the whole graph has been
cleaned.  \end{proof}

We will now show how to implement the targets using the tree-filling operation
defined below. 

\begin{definition}[Tree-filling] Let $G=(V,E)$ be a graph. Suppose that $S =
\{s_1, s_2, \ldots, s_j\}\subset V$ is a set of vertices with target $t$.  The
\emph{tree-filling} operation is the following. First, we add in $G$ a binomial
tree $T_i$, where $i = \lceil\log j\rceil + t +1$. Observe that, by
Proposition~\ref{subtrees}, there exist $2^{i-t-1} > j$ vertex-disjoint and
non-adjacent sub-trees $T_t$ in $T_i$.  For each $s\in S$, we find such a copy
of $T_t$ in $T_i$, identify $s$ with its root $r_t$, and delete all other
vertices of the sub-tree $T_t$.  \end{definition}

The tree-filling operation might in general increase treewidth, but we will do
it in a way such that treewidth only increases by a constant factor compared to
the pathwidth of $G$.

\begin{lemma}\label{lemma:pwtotw} Let $G=(V,E)$ be a graph of pathwidth $w$ and
$S=\{s_1,\ldots,s_j\}\subset V$ a subset of vertices having target $t$.  Then
there is a way to apply the tree-filling operation such that the resulting
graph $H$ has $tw(H) \le 4w +5$.  \end{lemma}

\begin{proof}

\textbf{Construction of $H$.} Let $(\mathcal{P},\mathcal{B})$ be a
path-decomposition of $G$ whose largest bag has size $w+1$ and $B_1, B_2,
\ldots, B_j\in \mathcal{B}$ distinct bags where $\forall a, s_a\in B_a$
(assigning a distinct bag to each $s_a$ is always possible, as we can duplicate
bags if necessary). We call those bags \emph{important}. We define an ordering
$o:S\rightarrow$ I\!N of the vertices of $S$ that follows the order of the
important bags from left to right, that is $o(s_a) < o(s_b)$ if $B_a$ is on the
left of $B_b$ in $\mathcal{P}$. For simplicity, let us assume that $o(s_a) = a$
and that $B_a$ is to the left of $B_b$ if $a<b$. 

We describe a recursive way to do the substitution of the trees in the
tree-filling operation. Crucially, when $j>2$ we will have to select an
appropriate mapping between the vertices of $S$ and the disjoint subtrees $T_t$
in the added binomial tree $T_i$, so that we will be able to keep the treewidth
of the new graph bounded.

\begin{itemize} 

\item If $j=1$ then $i = t+1$. We add to the graph a copy of $T_i$, arbitrarily
select the root of a copy of $T_t$ contained in $T_i$, and perform the
tree-filling operation as described.

\item Suppose that we know how to perform the substitution for sets of size at
most $\lceil j/2\rceil$, we will describe the substitution process for a set of
size $j$. We have $i=\lceil\log j \rceil +t+1$ and for all $j$ we have $\lceil
\log \lceil j/2\rceil \rceil = \lceil \log j \rceil -1$.  Split the set $S$
into two (almost) equal disjoint sets $S^L$ and $S^R$ of size at most $\lceil
j/2\rceil$, where for all $s_a\in S^L$ and for all $s_b\in S^R$, $a<b$. We
perform the tree-filling on each of these sets by constructing two binomial
trees $T_{i-1}^L, T_{i-1}^R$ and doing the substitution; then, we connect their
roots and set the root of the left tree as the root $r_i$ of $T_i$, thus
creating the substitution of a tree $T_i$.  \end{itemize}

\textbf{Small treewidth.} We now prove that the new graph $H$ that results from
applying the tree-filling operation on $G$ and $S$ as described above has a
tree decomposition $(\mathcal{T},\mathcal{B'})$ of width $4w+5$; in fact we
prove by induction on $j$ a stronger statement: if $A,Z\in \mathcal{B}$ are the
left-most and right-most bags of $\mathcal{P}$, then there exists a tree
decomposition $(\mathcal{T},\mathcal{B'})$ of $H$ of width $4w+5$ with the
added property that there exists $R\in\mathcal{B'}$ such that $A\cup
Z\cup\{r_i\}\subset R$, where $r_i$ is the root of the tree $T_i$. 

For the base case, if $j=1$ we have added to our graph a $T_i$ of which we have
selected an arbitrary sub-tree $T_t$, and identified the root $r_t$ of $T_t$
with the unique vertex of $S$ that has a target. Take the path decomposition
$(\mathcal{P},\mathcal{B})$ of the initial graph and add all vertices of $A$
(its first bag) and the vertex $r_i$ (the root of $T_i$) to all bags. Take an
optimal tree decomposition of $T_i$ of width $1$ and add $r_i$ to each bag, obtaining a
decomposition of width $2$. We add an edge between the bag of $\mathcal{P}$
that contains the unique vertex of $S$, and a bag of the decomposition of $T_i$
that contains the selected $r_t$. We now have a tree decomposition of the new
graph of width $2w+2<4w+5$. Observe that the last bag of $\mathcal{P}$ now
contains all of $A,Z$ and $r_i$.

For the inductive step, suppose we applied the tree-filling operation for a set
$S$ of size $j>1$. Furthermore, suppose we know how to construct a tree
decomposition with the desired properties (width $4w+5$, one bag contains the
first and last bags of the path decomposition $\mathcal{P}$ and $r_i$), if we
apply the tree-filling operation on a target set of size at most $j-1$. We show
how to obtain a tree decomposition with the desired properties if the target
set has size $j$.

By construction, we have split the set $S$ into two sets $S^L,S^R$ and have applied the
tree-filling operation to each set separately. Then, we connected the roots of
the two added trees to obtain a larger binomial tree. Observe that for
$|S|=j>1$ we have $|S^L|,|S^R|<j$.

Let us first cut $\mathcal{P}$ in two parts, in such a way that the important
bags of $S^L$ are on the left and the important bags of $S^R$ are on the right.
We call $A^L = A$ and $Z^L$ the leftmost and rightmost bags of the left part
and $A^R$, $Z^R = Z$ the leftmost and rightmost bags of the right part. We
define as $G^L$ (respectively $G^R$) the graph that contains all the vertices
of the left (respectively right) part. Let $r_i$ be the root of $T_i$ and
$r_{i-1}$ the root of its subtree $T_{i-1}$. From the inductive hypothesis, we
can construct tree decompositions $(\mathcal{T^L,B^L})$, $(\mathcal{T^R,B^R})$
of width $4w+5$ for the graphs $H^L$, $H^R$ that occur after applying
tree-filling on $G^L, S^L$ and $G^R, S^R$; furthermore, there exist $R^L \in
\mathcal{B^L}, R^R\in \mathcal{B^R}$ such that $R^L \supseteq A \cup Z^L \cup
\{r_i\}$ and $R^R \supseteq A^R \cup Z \cup \{r_{i-1}\}$. 

We construct a new bag $R' = A \cup A^R \cup Z^L \cup Z \cup \{r_{i-1}, r_i\}$,
and we connect $R'$ to both $R^L$ and $R^R$, thus combining the two
tree-decompositions into one. Last we create a bag $R = A\cup Z \cup \{r_i\}$
and attach it to $R'$. This completes the construction of $(\mathcal{T,B'})$.

Observe that $(\mathcal{T,B'})$ is a valid tree-decomposition for $H$: 

\begin{itemize} 

\item $V(H) = V(H^L) \cup V(H^R)$, thus $\forall v\in V(H), v\in \mathcal{B^L}
\cup \mathcal{B^R} \subset \mathcal{B}$.  

\item $E(H) = E(H^L) \cup E(H^R) \cup \{(r_{i-1}, r_i)\}$. We have that
$r_{i-1}, r_i\in R'\in \mathcal{B}$. All other edges were dealt with in
$\mathcal{T^L, T^R}$.  

\item Each vertex $v\in V(H)$ that belongs in exactly one of $H^L, H^R$
trivially satisfied the connectivity requirement: bags that contain $v$ are
either fully contained in $\mathcal{T^L}$ or $\mathcal{T^R}$. A vertex $v$ that
is in both $H^L$ and $H^R$ is also in $Z^L\cap A^R$ due to the 
properties of path-decompositions, hence in $R'$.
Therefore, the sub-trees of bags that contain $v$ in $\mathcal{T^L,T^R}$, form
a connected sub-tree in $\mathcal{T}$.
\end{itemize} 

The width of $\mathcal{T}$ is $\max\{tw(H^L), tw(H^R), |R'|-1\} = 4w + 5$.
\end{proof}

The last thing that remains to do in order to complete the proof is to show the equivalence between achieving the targets and finding a Grundy coloring.

\begin{lemma}\label{lemma:thm} Let $G$ and $G'$ be two graphs as described in
Lemma~\ref{lemma:reduction-TGC} and let $H$ be constructed from $G'$ by using
the tree-filling operation. Then $G$ has a clique of size $k$ if and only if $\Gamma(H)
\ge \lceil\log(k(m+1)+{k\choose 2}+2m)\rceil  + 2\log n + 5$. Furthermore,
$tw(H) \le 4{k \choose 2} +8k +17$.  \end{lemma}

\begin{proof} We note that the number of vertices with targets in our
construction is $m'=k(m+1)+{k\choose 2}+2m$ (the propagators, edge selection
checkers, and edge-checkers).  From Lemma~\ref{lemma:reduction-TGC}, it only
suffices to show that $\Gamma(H) \ge \lceil \log m'\rceil + 2\log n + 5$ if and only if
the vertices with targets achieve color $t=2\log n +4$.

For the forward direction, once vertices with targets get the desirable colors,
the rest of the binomial tree of the tree-filling operation can be colored
optimally, starting from its leaves all the way up to its roots, which will get
color $i=\lceil\log m'\rceil + 2\log n + 5$. 

For the converse direction, observe that the only vertices having degree higher
than $2\log n +4$ are the edge-checkers and the vertices of the binomial tree
$H\setminus G'$. However, the edge-checkers connect to only one vertex of
degree higher than $2\log n +4$, that in the binomial tree. Thus no vertex of
$G'$ can ever get a color higher than $2\log n +6$ and the only way that
$\Gamma(H) \ge \lceil\log m'\rceil + 2\log n + 5$ is if the root of the
binomial tree of the tree-filling operation (the only vertex of high enough
degree) receives color $\lceil\log m'\rceil + 2\log n + 5$. For that to happen,
all the support-trees of this tree should be colored optimally, which proves
that the vertices with targets $2 \log n +4$ having substituted support trees
$T_{2 \log n +4}$ should achieve their targets.

In terms of the treewidth of $H$ we have the following:
Lemma~\ref{lemma:smallpw} says that $G'$ once we remove all the supporting
trees has pathwidth at most ${k \choose 2} + 2k +3$. Applying
Lemma~\ref{lemma:pwtotw} we get that $H$ where we have ignored the
tree-supports from $G'$ has treewidth at most $4\left({k \choose 2} +2k
+3\right) +5$. Adding back the tree-supports does not increase its treewidth.
\end{proof}

The main theorem of this section now immediately follows.

\begin{theorem}\label{thm:tw} \textsc{Grundy Coloring} parameterized by treewidth is W[1]-hard.
\end{theorem}

\section{FPT for pathwidth}\label{sec:pw}

In this section, we show that, in contrast to treewidth, \textsc{Grundy
Coloring} is FPT parameterized by pathwidth. This is achieved by a combination
of an algorithm for \grundy\ given by Telle and Proskurowski and a
combinatorial bound due to Dujmovic, Joret, and Wood. We first recall these
results below.

\begin{lemma}[\cite{DujmovicJW12}]\label{lem:pw-upper-bound} For every graph
$G$, $\Gamma(G) \leq 8\cdot (pw(G)+1)$.  \end{lemma}

\begin{lemma}[\cite{TelleP97}] There is an algorithm which solves \grundy\ in
time $O^*(2^{O(tw(G)\cdot\Gamma (G))})$.\end{lemma}

We thus get the following result.

\begin{theorem}\label{thm:pw}
\textsc{Grundy Coloring} can be solved in time $O^*(2^{O(pw(G)^2)})$.
\end{theorem}

\begin{proof} Since in all graphs $tw(G)\le pw(G)$ and by Lemma
\ref{lem:pw-upper-bound} $\Gamma(G)\le 8(pw(G)+1)$, we have
$tw(G)\cdot\Gamma(G) = O(pw(G)^2)$ and the algorithm of \cite{TelleP97} runs in
at most the stated time. \end{proof}

\section{NP-hardness for Constant Clique-width}\label{sec:cw}

In this section we prove that \textsc{Grundy Coloring} is NP-hard even for
constant clique-width via a reduction from \textsc{3-SAT}.  We use a similar
idea of adding supports as in Section~\ref{sec:tw}, but supports now will be
cliques instead of binomial trees. The support operation is defined as:

\begin{definition} Given a graph $G=(V,E)$, a vertex $u\in V$ and a set of
positive integers $S$, we define the \textbf{support} operation as follows: for
each $i\in S$, we add to $G$ a clique of size $i$ (using new vertices) and we
connect one arbitrary vertex of each such clique to $u$.  \end{definition}

When applying the support operation we will say that we support vertex $u$ with
set $S$ and we will call the vertices introduced supporting vertices. 
%Note that
%the support operation may also be defined when $S$ is a multi-set, but in our
%construction we will only need to use it when $S$ is a set.  
Intuitively, the
support operation ensures that the vertex $u$ may have at least one neighbor
with color $i$ for each $i\in S$.% in a Grundy coloring.

We are now ready to describe our construction. Suppose we are given a 3CNF
formula $\phi$ with $n$ variables $x_1,\ldots,x_n$ and $m$ clauses
$c_1,\ldots,c_m$.  We assume without loss of generality that each clause
contains exactly three variables.
%, since \textsc{3-SAT} is still NP-hard under
%this restriction.  
We construct a graph $G(\phi)$ as follows:

\begin{enumerate}

\item For each $i\in [n]$ we construct two vertices $x_i^P, x_i^N$ and the edge
$(x_i^P,x_i^N)$.

\item For each $i\in [n]$ we support the vertices $x_i^P, x_i^N$ with the set
$[2i-2]$. (Note that $x_1^P, x_1^N$ have empty support).

\item For each $i\in [n], j\in [m]$, if variable $x_i$ appears in clause $c_j$
then we construct a vertex $x_{i,j}$. Furthermore, if $x_i$ appears positive in
$c_j$, we connect $x_{i,j}$ to $x_{i'}^P$ for all $i'\in [n]$; otherwise we
connect $x_{i,j}$ to $x_{i'}^N$ for all $i'\in [n]$.

\item For each $i\in [n], j\in[m]$ for which we constructed a vertex $x_{i,j}$
in the previous step, we support that vertex with the set $(\{ 2k\ |\ k\in[n]\}
\cup \{2i-1, 2n+1, 2n+2\})\setminus \{2i\}$.

\item For each $j\in[m]$ we construct a vertex $c_j$ and connect to all (three)
vertices $x_{i,j}$ already constructed. We support the vertex $c_j$ with the
set $[2n]$.

\item For each $j\in[m]$ we construct a vertex $d_j$ and connect it to $c_j$.
We support $d_j$ with the set $[2n+3] \cup [2n+5, 2n+3+j]$.

\item We construct a vertex $u$ and connect it to $d_j$ for all $j\in[m]$. We
support $u$ with the set $[2n+4]\cup [2n+5+m, 10n+10m]$.

\end{enumerate}

This completes the construction. Before we proceed, let us give some intuition.
Observe that we have constructed two vertices $x_i^P, x_i^N$ for each variable.
The support of these vertices and the fact that they are adjacent, allow us to
give them colors $\{2i-1, 2i\}$. The choice of which gets the higher color
encodes an assignment to variable $x_i$. The vertices $x_{i,j}$ are now
supported in such a way that they can ``ignore'' the values of all variables
except $x_i$; for $x_i$, however, $x_{i,j}$ ``prefers'' to be connected to a
vertex with color $2i$ (since $2i-1$ appears in the support of $x_{i,j}$, but
$2i$ does not).  Now, the idea is that $c_j$ will be able to get color $2n+4$
if and only if one of its literal vertices $x_{i,j}$ was ``satisfied'' (has a
neighbor with color $2i$). The rest of the construction checks if all clause
vertices are satisfied in this way.

We now state the lemmata that certify the correctness of our reduction. 

\begin{lemma}\label{lem:cw1} If $\phi$ is satisfiable then $G(\phi)$ has a
Grundy coloring with $10n+10m+1$ colors. \end{lemma}

\begin{proof}

Consider a satisfying assignment of $\phi$. We first produce a coloring of the
vertices $x_i^P, x_i^N$ as follows: if $x_i$ is set to True, then $x_i^P$ is
colored $2i$ and $x_i^N$ is colored $2i-1$; otherwise $x_i^P$ is colored $2i-1$
and $x_i^N$ is colored $2i$. Before proceeding, let us also color the
supporting vertices of $x_i^P, x_i^N$: each such vertex belongs to a clique
which contains only one vertex with a neighbor outside the clique. For each
such clique of size $\ell$, we color all vertices of the clique which have no
outside neighbors with colors from $[\ell-1]$ and use color $\ell$ for the
remaining vertex. Note that the coloring we have produced so far is a valid
Grundy coloring, since each vertex $x_i^P, x_i^N$ has for each $c\in[2i-2]$ a
neighbor with color $c$ among its supporting vertices, allowing us to use
colors $\{2i-1, 2i\}$ for $x_i^P, x_i^N$. In the remainder, we will use similar
such colorings for all supporting cliques. We will only stress the color given
to the vertex of the clique that has an outside neighbor, respecting the
condition that this color is not larger than the size of the clique. Note that
it is not a problem if this color is strictly smaller than the size of the
clique, as we are free to give higher colors to internal vertices.

Consider now a clause $c_j$ for some $j\in[m]$. Suppose that this clause
contains the three variables $x_{i_1}, x_{i_2}, x_{i_3}$. Because we started
with a satisfying assignment, at least one of these variables has a value that
satisfies the clause, without loss of generality $x_{i_3}$. We therefore color
$x_{i_1}, x_{i_2}, x_{i_3}$ with colors $2n+1, 2n+2, 2n+3$ respectively and we
color $c_j$ with color $2n+4$. We now need to show that we can appropriately
color the supporting vertices to make this a valid Grundy coloring.

Recall that the vertex $x_{i_3}$ has support $\{2,4,\ldots, 2n\}\setminus
\{2i_3\} \cup \{2i_3-1, 2n+1, 2n+2\}$. For each $i'\neq i_3$ we observe that
$x_{i_3}$ is connected to a vertex (either $x_{i_3}^P$ or $x_{i_3}^N$) which
has a color in $\{2i'-1, 2i'\}$, we are therefore missing the other color from
this set.  We consider the clique of size $2i'$ supporting $x_{i_3,j}$: we
assign this missing color to the vertex of this clique that is adjacent to
$x_{i_3,j}$. Note that the clique is large enough to color its remaining
vertices with lower colors in order to make this a valid Grundy coloring. For
$i_3$, we observe that, since $x_{i_3}$ satisfies the clause, the vertex
$x_{i_3,j}$ has a neighbor (either $x_{i_3}^P$ or $x_{i_3}^N$) which has
received color $2i_3$; we use color $2i_3-1$ in the support clique of the same
size. Similarly, we use colors $2n+1, 2n+2$ in the support cliques of the same
sizes, and $x_{i_3}$ has neighbors with colors covering all of $[2n+2]$.

For the vertex $x_{i_2,j}$ we proceed in a similar way. For $i'< i_2$ we give
the support vertex from the clique of size $2i'$ the color from $\{2i'-1,
2i'\}$ which does not already appear in the neighborhood of $x_{i_2,j}$. For
$i'\in [i_2, n-1]$ we take the vertex from the clique of size $2i'+2$ and give
it the color of $\{2i'-1, 2i'\}$ which does not yet appear in the neighborhood
of $x_{i_2,j}$. In this way we cover all colors in $[2n-2]$. We now observe
that $x_{i_2,j}$ has a neighbor with color in $\{2n-1, 2n\}$ (either $x_n^P$ or
$x_n^N$); together with the support vertices from the cliques of sizes $2n+1,
2n+2$ this allows us to cover the colors $[2n-1, 2n+1]$. We use a similar
procedure to cover the colors $[2n]$ in the neighborhood of $x_{i_1,j}$. Now,
the $2n$ support vertices in the neighborhood of $c_j$, together with
$x_{i_1,j}, x_{i_2,j}, x_{i_3,j}$ allow us to give that vertex color $2n+4$.

We now give each vertex $d_j$, for $j\in[m]$ color $2n+j+4$. This can be
extended to a valid coloring, because $d_j$ is adjacent to $c_j$, which has
color $2n+4$, and the support of $d_j$ is $[2n+j+3]\setminus \{2n+4\}$.

Finally, we give $u$ color $10n+10m+1$. Its support is $[10n+10m]\setminus
[2n+5, 2n+m+4]$. However, $u$ is adjacent to all vertices $d_j$, whose colors
cover the set $\{ 2n+4+j\ |\ j\in[m]\}$.  \end{proof}

\begin{lemma}\label{lem:cw2} If $G(\phi)$ has a Grundy coloring with
$10n+10m+1$ colors, then $\phi$ is satisfiable. \end{lemma}

\begin{proof}

Consider a Grundy coloring of $G(\phi)$. We first assume without loss of
generality that we consider a minimal induced subgraph of $G$ for which the
coloring remains valid, that is, deleting any vertex will either reduce the
number of colors or invalidate the coloring. In particular, this means there is
a unique vertex with color $10n+10m+1$. This vertex must have degree at least
$10n+10m$. However, there are only two such vertices in our graph: $u$ and its
support neighbor vertex in the clique of size $10n+10m$. If the latter vertex
has color $10n+10m+1$, we can infer that $u$ has color $10n+10m$: this color
cannot appear in the clique because all its internal vertices have degree
$10n+10m-1$, and one of their neighbors has a higher color. We observe now that
exchanging the colors of $u$ and its neighbor produces another valid coloring.
We therefore assume without loss of generality that $u$ has color $10n+10m+1$.

We now observe that in each supporting clique of $u$ of size $i$ the maximum
color used is $i$ (since $u$ has the largest color in the graph). Similarly,
the largest color that can be assigned to $d_j$ is $2n+j+4$, because $d_j$ has
degree $2n+j+4$, but one of its neighbors ($u$) has a higher color. We conclude
that the only way for the $10n+10m$ neighbors of $u$ to cover all colors in
$[10n+10m]$ is for each support clique of size $i$ to use color $i$ and for
each $d_j$ to be given color $2n+j+4$.

Suppose now that $d_j$ was given color $2n+j+4$. This implies that the largest
color that $c_j$ may have received is $2n+4$, since its degree is $2n+4$, but
$d_j$ received a higher color. We conclude again that for the neighbors of
$d_j$ to cover $[2n+j+3]$ it must be the case that each supporting clique used
its maximum possible color and $c_j$ received color $2n+4$.

Suppose now that a vertex $c_j$ received color $2n+4$. Since $d_j$ received a
higher color, the remaining $2n+3$ neighbors of this vertex must cover
$[2n+3]$. In particular, since the support vertices have colors in $[2n]$, its
three remaining neighbors, say $x_{i_1,j}, x_{i_2,j}, x_{i_3,j}$ must have
colors covering $[2n+1, 2n+3]$. Therefore, all vertices $x_{i,j}$ have colors
in $[2n+1, 2n+3]$.

Consider now two vertices $x_i^P, x_i^N$, for some $i\in [n]$. We claim that
the vertex which among these two has the lower color, has color at most $2i-1$.
To see this observe that this vertex may have at most $2i-2$ neighbors from the
support vertices that have lower colors and these must use colors in $[2i-2]$
because of their degrees. Its neighbors of the form $x_{i,j}$ have color at
least $2n+1> 2i-1$, and its neighbor in $\{x_i^P, x_i^N\}$ has a higher color.
Therefore, the smaller of the two colors used for $\{x_i^P, x_i^N\}$ is at most
$2i-1$ and by similar reasoning the higher of the two colors used for this set
is at most $2i$. We now obtain an assignment for $\phi$ by setting $x_i$ to
True if $x_i^P$ has a higher color than $x_i^N$ and False otherwise (this is
well-defined, since $x_i^P, x_i^N$ are adjacent). 

Let us argue why this is a satisfying assignment. Take a clause vertex $c_j$.
As argued, one of its neighbors, say $x_{i_3,j}$ has color $2n+3$. The degree
of $x_{i_3,j}$, excluding $c_j$ which has a higher color, is $2n+2$, meaning
that its neighbors must exactly cover $[2n+2]$ with their colors. Since
vertices $x_i^P, x_i^N$ have color at most $2i$, the colors $[2n+1,2n+2]$ must
come from the support cliques of the same sizes. Now, for each $i\in[n]$ the
vertex $x_{i_3,j}$ has exactly two neighbors which may have received colors in
$\{2i-1, 2i\}$. This can be seen by induction on $i$: first, for $i=n$ this is
true, since we only have the support clique of size $2n$ and the neighbor in
$\{x_n^P, x_n^N\}$. Proceeding in the same way we conclude the claim for
smaller values of $i$. The key observation is now that the clique of size
$2i_3-1$ cannot give us color $2i_3$, therefore this color must come from
$\{x_{i_3}^N, x_{i_3}^P\}$. If the neighbor of $x_{i_3,j}$ in this set uses
$2i_3$, this must be the higher color in this set, meaning that $x_{i_3}$ has a
value that satisfies $c_j$.  \end{proof}

\begin{lemma}\label{lem:cw3} The graph $G(\phi)$ has clique-width at most 8.
\end{lemma}

\begin{proof}[Lemma \ref{lem:cw3}]

Let us first observe that the support operation does not significantly affect a
graph's clique-width. Indeed, if we have a clique-width expression for
$G(\phi)$ without the support vertices, we can add these vertices as follows:
each time we introduce a vertex that must be supported we instead construct the
graph induced by this vertex and its support and then rename all supporting
vertices to a junk label that is never connected to anything else. It is clear
that this can be done by adding at most three new labels: two labels for
constructing the clique (that will form the support gadget) and the junk label.
In fact, below we give a clique-width expression for the rest of the graph that
already uses a junk label (say, label $0$), that is, a label on which we never
apply a Join operation.  Hence, it suffices to compute the clique-width of
$G(\phi)$ without the support gadgets and then add $2$.

Let us then argue why the rest of the graph has constant clique-width. First,
the graph induced by $x_i^N, x_i^P$, for $i\in[n]$ is a matching. We construct
this graph using $4$ labels, say $1,2,3,4$ as follows: for each $i\in [n]$ we
introduce $x_i^N$ with label $3$, $x_i^P$ with label $4$, perform a Join
between labels $3$ and $4$, then Rename label $3$ to $1$ and label $4$ to $2$.
This constructs the matching induced by these $2n$ vertices and also ensures
that all vertices $x_i^N$ have label $1$ in the end and all vertices $x_i^P$
have label $2$ in the end.

We then introduce to the graph the clauses one by one. Specifically, for each
$j\in [m]$ we do the following: we introduce $c_j$ with label $3$, $d_j$ with
label $4$, Join labels $3$ and $4$, Rename label $4$ to label $5$; then for
each $i\in [n]$ such that we have a vertex $x_{i,j}$ we introduce that vertex
with label $4$, Join label $4$ with label $3$, and Join label $4$ with label
$1$ or $2$, depending on whether $x_{i,j}$ is connected to vertices $x_i^N$ or
$x_i^P$, then Rename label $4$ to the junk label $0$. Once all $x_{i,j}$
vertices for a fixed $j$ have been introduced we Rename label $3$ to the junk
label $0$ and move to the next clause. Finally, we introduce $u$ with label $3$
and Join label $3$ to label $5$ (which is the label shared by all $d_j$
vertices). In the end we have used $6$ labels, namely the labels
$\{0,1,2,3,4,5\}$ for $G(\phi)$ without the support vertices, so the whole
graph can be constructed with $8$ labels. \end{proof}

\begin{theorem}\label{thm:cw}
Given graph $G=(V,E)$, $k$-\textsc{Grundy Coloring} is NP-hard even when the
clique-width of the graph $cw(G)$ is a fixed constant.  \end{theorem}

\section{FPT for modular-width}\label{sec:mw}

In this section we show that \textsc{Grundy Coloring} is FPT parameterized by
modular-width. Recall that $G=(V,E)$ has modular-width $w$ if $V$ can be
partitioned into at most $w$ modules, such that each module is a singleton or
induces a graph of modular-width $w$. Neighborhood diversity is the restricted
version of this measure where modules are required to be cliques or independent
sets. 

The first step is to show that \textsc{Grundy Coloring} is FPT parameterized by
neighborhood diversity. Similarly to the standard \textsc{Coloring} algorithm
for this parameter \cite{Lampis12}, we observe that, without loss of
generality, all modules can be assumed to be cliques, and hence any color class
has one of $2^w$ possible types, depending on the modules it intersects.  We
would like to use this to reduce the problem to an ILP with $2^w$ variables,
but unlike \textsc{Coloring},  the ordering of color classes matters.  We thus
prove that the optimal solution can be assumed to have a ``canonical''
structure where each color type only appears in consecutive colors. 
%This allows us to guess the structure of the optimal solution and then reduce
%the problem to ILP. 
We then extend the neighborhood diversity algorithm to modular-width using the
idea that we can calculate the Grundy number of each module separately, and
then replace it with an appropriately-sized clique. 
%Once we show that this transformation is sound, we use it to reduce the
%problem to the neighborhood diversity case. 

\subsection{Neighborhood diversity}\label{subsec:nd} 

Recall that two vertices $u,v\in V$ of a graph $G=(V,E)$ are \emph{twins} if
$N(u)\setminus \{v\} = N(v)\setminus \{u\}$, and they are  called \emph{true}
(respectively, \emph{false}) twins if they are adjacent (respectively,
non-adjacent). A \emph{twin class} is a maximal set of vertices that are
pairwise twins.  It is easy to see that any  twin class is either a clique or
an independent set.  We say that a graph $G=(V,E)$ has \emph{neighborhood
diversity} $w$ if $V$ can be partitioned into at most $w$ twin classes. 

Let $G=(V,E)$ be a graph of neighborhood diversity $w$ with a vertex partition
$V=W_1\dot{\cup} \ldots \dot{\cup} W_w$ into twin classes.  It is obvious that
in any Grundy Coloring of $G$, the vertices of a true twin class must have all
distinct colors because they form a clique. Furthermore, it is not difficult to
see that the vertices of a false twin class must be colored by the same color
because all of their vertices have the same neighbors. 

In fact, we can show that we can remove vertices from a false twin class
without affecting the Grundy number of the graph:

\begin{lemma}\label{lem:falsesingle} Let $G=(V,E)$ be a graph of neighborhood
diversity $w$ with a vertex partition $V=W_1\dot{\cup} \ldots \dot{\cup} W_w$
into twin classes.  Let $W_i$ be a false twin class having at least two
distinct vertices $u,v\in W_i$.  Then $G- v$ has $k$-Grundy coloring if and
only if $G$ has.  \end{lemma} 

\begin{proof} The forward implication is trivial.  To see the opposite
direction, consider an arbitrary $k$-Grundy coloring of $G$. The vertices $u,v$
must have the same color, since they have the same neighbors. Any vertex whose
color is higher than $v$ and is adjacent with $v$ must be to $u$ as well.
Since $u$ and $v$ have the same color, this implies that the same coloring
restricted to $G-v$ is a $k$-Grundy coloring.  \end{proof}

Using Lemma~\ref{lem:falsesingle}, we can reduce every false twin class into a
singleton vertex, thus from now on we may assume that every twin class is a
clique (possibly a singleton).  An immediate consequence is that that any color
class of a Grundy coloring can take at most one vertex from each twin class.
Furthermore, the colors of any two vertices from the same twin class are
interchangeable.  Therefore, a color class $V_i$ of a Grundy coloring is
precisely characterized by the set of twin classes $W_j$ that $V_i$ intersects.
For a color class $V_i$, we call the set $\{j\in [w]: W_j\cap V_i\neq
\emptyset\}$ as the \emph{intersection pattern} of $V_i$.

Let $\cal{I}$ be the collection of all sets $I\subseteq [w]$ of indices such
that $W_i$ and $W_j$ are non-adjacent for every distinct pairs $i,j\in [w]$.
It is clear that the intersection pattern of any color class is a member of
$\cal{I}$. It turns out that if $I\in \cal{I}$ appears as an intersection
pattern for more than one color classes, then it can be assumed to appear on a
consecutive set of colors.

\begin{lemma}\label{lem:moveint}
Let $G=(V,E)$ be a graph of neighborhood diversity $w$ with a vertex partition $V=W_1\dot{\cup} \ldots \dot{\cup} W_w$ into true twin classes. 
Let $V_1\dot{\cup} \ldots \dot{\cup} V_k$ be a $k$-Grundy coloring of $G$ and 
let $I_i\in \cal{I}$ be the set of indices $j$ such that $V_i\cap W_j\neq \emptyset$ for each $i\in [k]$. 
If $I_i=I_{i'}$ for some $i'\geq i+2$, then the coloring $V'_1\dot{\cup} \ldots \dot{\cup} V'_k$ where
\begin{equation*}
V'_{\ell}=
\begin{cases}
V_{i'} &\text{if }  \ell = i+1,\\
V_{\ell-1} &\text{if } i+1< \ell \le i',\\
V_{\ell} &\text{otherwise}
\end{cases}
\end{equation*}
(i.e. the coloring obtained by `inserting' $V_{i'}$ in between $V_i$ and $V_{i+1}$) is a Grundy coloring as well.
\end{lemma}

\begin{proof}

First observe that the new coloring remains a proper coloring, so we only need
to argue that it's a valid Grundy coloring. Consider a vertex $v$ which took
color $j\le i$ in the original coloring. All its neighbors with color strictly
smaller than $j$ have retained their colors, so $v$ is still properly colored.
Suppose then that $v$ had color $j>i'$ in the original coloring. Then, $v$ has
a neighbor in each of the classes $V_1,\ldots, V_{j-1}$, which means that it
has at least one neighbor in each of the sets $V'_1,\ldots, V'_{j-1}$, so it is
still validly colored.

Suppose that $v$ had received a color $j\in [i+1,i'-1]$ in the original
coloring and receives color $j+1$ in the new coloring. We claim that for each
$j'<j+1$, $v$ has a neighbor with color $j'$. Indeed, this is easy to see for
$j'\le i$, as these vertices retain their colors; for $j'=i+1$ we observe that
$v$ has a neighbor with color $i$ in the original coloring, and each such
vertex has a true twin with color $i+1$ in the new coloring; and for $j'>i+1$,
the neighbor of $v$ which had color $j'-1$ originally now has color $j'$.

Finally, suppose that $v$ had received color $i'$ in the original coloring and
receives color $i+1$ in the new coloring. We now observe that such a vertex $v$
must have a true twin which received color $i$ in both colorings, therefore
coloring $v$ with $i+1$ is valid.  \end{proof}

The following is a consequence of Lemma~\ref{lem:moveint}.

\begin{corollary}\label{corr:intpartition} Let $G=(V,E)$ be a graph of
neighborhood diversity $w$ with a vertex partition $V=W_1\dot{\cup} \ldots
\dot{\cup} W_w$ into true twin classes.  If $G$ admits a $k$-Grundy coloring,
then there is a $k$-Grundy coloring $V_1\dot{\cup} \ldots \dot{\cup} V_k$ with
the following property: for each $j_1,j_2\in [k]$ such that $V_{j_1}$ has a
non-empty intersection with the same twin classes as $V_{j_2}$, we have that
for all $j_3\in [k]$ with $j_1\le j_3\le j_2$, $V_{j_3}$ also has non-empty
intersection with the same twin classes as $V_{j_1}$.  \end{corollary}

For a sub-collection $\cal{I}'$ of $\cal{I}$, we say that $\cal{I}'$ is
\emph{eligible} if there is an ordering $\preceq$ on $\cal{I}'$ such that for
every $I,I'\in \cal{I}'$ with $I\succeq I'$, and for every $i\in I$, there
exists $i'\in I'$ such that the twin classes $W_i$ and $W_{i'}$ are adjacent,
or $i=i'$.  Clearly, a sub-collection of an eligible sub-collection of
$\cal{I}$ is again eligible. Intuitively, the ordering that shows that a
sub-collection is eligible corresponds to a Grundy coloring where color classes
have the corresponding intersection patterns.

Now we are ready to present an FPT algorithm, parameterized by the neighborhood
diversity $w$, to compute the Grundy number.  The algorithm consists of two
steps: (i) guess a sub-collection $\cal{I}'$ of $\cal{I}$ which are used as
intersection patterns by a Grundy coloring, and  (ii) given $\cal{I}'$, we
solve an integer linear program. 

Let $\cal{I}'$ be a sub-collection of $\cal{I}$.  For each $I\in \cal{I}'$, let
$x_I$ be an integer variable which is interpreted as the number of colors for
which $I$ appears as an intersection pattern.  Now, the linear integer program
\texttt{ILP}($\cal{I}'$) for a sub-collection $\cal{I}'$ is given as the
following: 

\begin{eqnarray} 
\max \sum_{I\in \cal{I}'} x_I &&\\ 
\text{s.t.} \qquad && \nonumber\\
&&\sum_{I\in {\cal I}': i\in I} x_I = |W_i| \qquad \forall i\in [w],
\end{eqnarray} 

where each $x_I$ takes a positive integer value.

\begin{lemma}\label{lem:maxgood}
Let $G=(V,E)$ be a graph of neighborhood diversity $w$ with a vertex partition $V=W_1\dot{\cup} \ldots \dot{\cup} W_w$ into true twin classes. 
The maximum value of \texttt{ILP}($\cal{I}'$) over all eligible $\cal{I}'\subseteq \cal{I}$ equals the Grundy number of $G$.
\end{lemma}

\begin{proof} We first prove that the maximum value over all considered
\texttt{ILP}s is at least the Grundy number of $G$.  Fix a Grundy coloring
$V_1\dot{\cup} \cdots \dot{\cup} V_k$ achieving the Grundy number while
satisfying the condition of Corollary~\ref{corr:intpartition}.  Consider the
sub-collection $\cal{I}'$ of $\cal I$ used as intersection patterns in the
fixed Grundy coloring.  It is clear that $\cal{I}'$ is eligible, using the
natural ordering of the color classes.  Let $\bar{x}_I$ be the number of colors
for which $I$ is an intersection pattern for each $I\in \cal{I}'$.  It is
straightforward to check that setting the variable $x_I$ at value $\bar{x}_I$
satisfies the constraints of \texttt{ILP}($\cal{I}'$), because all vertices of
each twin class are colored exactly once.  Therefore, the objective value of
\texttt{ILP}($\cal{I}'$) is at least the Grundy number.

To establish the opposite direction of inequality, let $\cal{I}'$ be an eligible sub-collection of $\cal{I}$ achieving the 
maximum ILP objective value. Notice that \texttt{ILP}($\cal{I}'$) is feasible, and let $x^*_I$ be the value taken by 
the variable $x_I$ for each $I\in \cal{I}'$. 
Since $\cal{I}'$ is eligible, there exists an ordering $\preceq$ on $\cal{I}'$ such that for every $I,I'\in \cal{I}'$ with $I\succeq I'$, and for every $i\in I$, there exists $i'\in I'$ 
such that the twin classes $W_i$ and $W_{i'}$ are adjacent. 
Now, we can define the coloring $V_1\dot{\cup} \cdots \dot{\cup} V_{\ell}$ by taking 
the first (i.e. minimum element in $\preceq$) element $I_1$ of $\cal I'$ $x^*_I$ times. 
That is, each of $V_1$ up to $V_{x^*_{I_1}}$ contains precisely one vertex of $W_i$ for each $i\in I$. 
The succeeding element $I_2$ similarly yields the next $x^*_{I_2}$ colors, and so on. 
From the constraint of \texttt{ILP}($\cal{I}'$), we know that the constructed coloring indeed partitions $V$. 
The eligibility of $\cal I'$ ensure that this is a Grundy coloring. Finally, observe that 
the number of colors in the constructed coloring equals the objective value of \texttt{ILP}($\cal{I}'$). 
This proves that the latter value is the lower bound for the Grundy number.
\end{proof}

\begin{theorem}\label{thm:grundynd} Let $G=(V,E)$ be a graph of neighborhood
diversity $w$.  The Grundy number of $G$ can be computed in time
$2^{O(w2^{w})}n^{O(1)}$.    \end{theorem} 

\begin{proof} We first compute the partition $V=W_1\dot{\cup} \ldots \dot{\cup}
W_w$ of $G$ into twin classes in polynomial time.  By
Lemma~\ref{lem:falsesingle}, we may assume that each $W_i$ is a true twin class
by discarding some vertices of $G$, if necessary.  Next, we compute $\cal I$
and notice that $\cal I$ contains at most $2^w$ elements. For each $\cal
I'\subseteq \cal I$ we verify if $\cal I'$ is eligible (this can be done in by
trying all $w!$ orderings of the elements of $\cal I'$).

For each eligible sub-collection of $\cal I'$ of $\cal I$, we solve
\texttt{ILP}($\cal{I}'$) using Lenstra's algorithm which runs in time
$O(n^{2.5n +o(n)})$, where $n$ denotes the number of variables in a given
linear integer program \cite{Lenstra83,Kannan87,FrankT87}. As
\texttt{ILP}($\cal{I}'$) contains as many as $|{\cal I}'| \leq 2^w$ variables,
this lead to an ILP solver running in time $2^{O(w2^w)}$.  Due to
Lemma~\ref{lem:maxgood}, we can correctly compute the Grundy number by solving
\texttt{ILP}($\cal{I}'$) for each eligible $\cal I'$ and taking the maximum.
\end{proof}

\subsection{Modular-width}

Let $G=(V,E)$ be a graph. A \emph{module} is a set $X\subseteq V$ of vertices
such that $N(u)\setminus X=N(v)\setminus X$ for every $u,v\in X$, that is,
their neighborhoods coincide outside of $X$. Equivalently, $X$ is a module if
all vertices of $V\setminus X$ are either connected to all vertices of $X$ or
to none. The modular width of a graph $G=(V,E)$ is defined recursively as
follows: (i) the modular width of a singleton vertex is $1$ (ii) $G$ has
modular width at most $k$ if and only if there exists a partition
$V=V_1\dot{\cup} \ldots \dot{\cup} V_k$, such that for all $i\in[k]$, $V_i$ is
a module and $G[V_i]$ has modular width at most $k$.

Our main tool in this section will be the following lemma which will allow us
to reduce \grundy\ parameterized by modular width to the same problem
parameterized by neighborhood diversity. We will then be able to invoke Theorem
\ref{thm:grundynd}. The idea of the lemma is that once we compute the Grundy
number of a module of a graph $G$ we can remove it and replace it with an
appropriately sized clique without changing the Grundy number of $G$.

\begin{lemma}\label{lem:mw} Let $G=(V,E)$ be a graph and $X\subseteq S$ be a
module of $G$.  Let $G'$ be the graph obtained by deleting $X$ from $G$ and
replacing it with a clique $X'$ of size $\Gamma(G[X])$, such that in $G'$ we
have that all vertices of $X'$ are connected to all neighbors of $X$ in $G$.
Then $\Gamma(G)= \Gamma(G')$. \end{lemma}

\begin{proof}

Let $k=\Gamma(G[X])=|X'|$. First, let us show that $\Gamma(G')\ge \Gamma(G)$.
Take a Grundy coloring of $G$. Our main observation is that the vertices of $X$
are using at most $k$ distinct colors in the coloring of $G$. To see this,
suppose for contradiction that the vertices of $X$ are using at least $k+1$
colors. We will show how to obtain a Grundy coloring of $G[X]$ with at least
$k+1$ colors. As long as there is a color in the Grundy coloring of $G$ which
does not appear in $X$, let $c$ be the highest such color. We delete from $G$
all vertices which have color $c$, and decrease by $1$ the color of all
vertices that have color greater than $c$. This modification gives us a valid
Grundy coloring of the remaining graph, without decreasing the number of
distinct colors used in $X$. Repeating this exhaustively results in a graph
where every color is used in $X$. Since $X$ is a module, that means that the
resulting graph is $G[X]$, and we have obtained a Grundy coloring of $G[X]$
with $k+1$ or more colors, contradiction.

Assume then that in the optimal Grundy coloring of $G$, the vertices of $X$ use
$k'\le k$ distinct colors. Let $G''$ be the induced subgraph of $G'$ obtained
by deleting vertices of $X'$ so that there are exactly $k'$ such vertices left
in the graph. We claim $\Gamma(G')\ge \Gamma(G'')\ge \Gamma (G)$. The first
inequality follows from the standard fact that Grundy coloring is closed under
induced subgraphs (indeed, in the First-Fit formulation of the problem we can
place the deleted vertices of $G'$ at the end of the ordering). To see that
$\Gamma(G'')\ge \Gamma(G)$ we take the optimal coloring of $G$ and use the same
coloring in $V\setminus X$; furthermore, for each distinct color used in a
vertex of $X$ we color a vertex of $X'$ with this color. Observe that this is a
proper coloring of $G''$. Furthermore, for each $v\in V\setminus X$, the set of
colors that appears in $N(v)$ is unchanged; while for $v\in X'$, $v$ sees at
least the same colors in its neighborhood as a vertex of $X$ that received the
same color.

Let us also show that $\Gamma(G) \ge \Gamma(G')$. Consider a $k$-Grundy
coloring of $G[X]$ and let $X_1, X_2, \ldots, X_k$ be the corresponding
partition of $X$. Label the vertices of $X'$ as $x_1,\ldots, x_k$. We will now
show how to transform a Grundy coloring of $G'$ to a Grundy coloring of $G$: we
use the same colors as in $G'$ for all vertices in $V\setminus X$; and we use
for each vertex of $X_i$ the same color that is used for $x_i$ in $G'$. This is
a proper coloring, as each $X_i$ is an independent set, the vertices of $X'$
use distinct colors in $G'$ (as they form a clique), and a vertex connected to
$X$ in $G$ is also connected to all of $X'$ in $G'$. Furthermore, each vertex
$v\in V\setminus X$ sees the same set of colors in its neighborhood in $G$ and
in $G'$: if $v$ is not connected to $X$ its neighborhood is completely
unchanged, while if it is $v$ sees in $X$ the same $k$ colors that were used in
$X'$. Finally, for each $i\in[k]$, each vertex of $X_i$ sees the same colors in
its neighborhood as $x_i$ does in $G'$.  \end{proof}

We can now prove the main result of this section.

\begin{theorem}\label{thm:mw} Let $G=(V,E)$ be a graph of modular-width $w$.
The Grundy number of $G$ can be computed in time $2^{O(w2^w)}n^{O(1)}$.
\end{theorem}

\begin{proof}

Given a graph $G=(V,E)$ of modular width $w$ it is known that we can compute a
partition of $V$ into at most $w$ modules $V_1,\ldots,V_w$ \cite{TedderCHP08}.
If one of these modules $V_i$ is not a clique or an independent set, we call
this algorithm recursively on $G[V_i]$ (which also has modular width $w$) and
compute $\Gamma(G[V_i])$. Then, by Lemma \ref{lem:mw} we can replace $V_i$ in
$G$ with a clique of size $\Gamma(G[V_i])$. Repeating this produces a graph
where each module is a clique or an independent set. But then $G$ has
neighborhood diversity $w$, so we can invoke Theorem \ref{thm:grundynd}.
\end{proof}

\section{Conclusions}

We have shown that \textsc{Grundy Coloring} is a natural problem that displays
an interesting complexity profile with respect to some of the main graph
widths. One question left open with respect to this problem is its complexity
parameterized by feedback vertex set.  A further question is the tightness of
our obtained results under the ETH. The algorithm we obtain for pathwidth has
running time with parameter dependence $2^{O(pw^2)}$. Is this optimal or is it
possible to do better? Similarly, our reduction for treewidth shows that it's
not possible to solve the problem is $n^{o(\sqrt{tw})}$, but the best known
algorithm runs in $n^{O(tw^2)}$. Can this gap be closed? 

A broader question is also whether we can find other examples of natural
problems that separate the parameters treewidth and pathwidth. The reason that
\grundy\ turns out to be tractable for pathwidth is purely combinatorial (the
value of the optimal is bounded by a function of the parameter). In other
words, the ``reason'' why this problem becomes easier for pathwidth is not that
we are able to formulate a different algorithm, but that the same algorithm
happens to become more efficient. It would be interesting to find some natural
problem for which pathwidth offers algorithmic footholds in comparison to
treewidth that cannot be so easily explained. One possible candidate for this
may be \textsc{Packing Coloring} \cite{KimLMP18}.

%%
%% Bibliography
%%

%% Please use bibtex, 
\bibliography{grundy}

\end{document}